\documentclass[preprint]{elsarticle}
\usepackage{amssymb,graphicx,bm}
\usepackage{amsmath,amsthm}
\usepackage{color}
\usepackage{caption, subcaption}
\usepackage{verbatim}
\newtheorem{lemma}{Lemma}
\newtheorem{corollary}{Corollary}

\newcommand{\p}{\partial}
\newcommand{\pdiff}[2]{\frac{\partial {#1}}{\partial {#2}}}
\newcommand{\Lag}{\mathcal{L}}

\newcommand{\alphab}{\boldsymbol{\alpha}}

\newcommand{\ba}{\begin{array}}
\newcommand{\ea}{\end{array}}
\newcommand{\be}{\begin{equation}}
\newcommand{\ee}{\end{equation}}
\newcommand{\bd}{\begin{displaymath}}
\newcommand{\ed}{\end{displaymath}}

\newcommand{\changed}{}
\newcommand{\changethree}{}

\begin{document}
\begin{frontmatter}
\title{Discrete Adjoints for Accurate Numerical Optimization with Application to Quantum Control}
\author[labelNAP]{N. Anders Petersson\corref{cor1}}
\ead{petersson1@llnl.gov}

\author[labelFMG]{Fortino M. Garcia}
\ead{fortino.garcia@colorado.edu}

\author[labelAEC]{Austin E. Copeland}
\ead{acopeland@mail.smu.edu}

\author[labelYJR]{Ylva~L.~Rydin}
\ead{ylva.rydin@it.uu.se}

\author[labelJLD]{Jonathan L. DuBois}
\ead{dubois9@llnl.gov}

\cortext[cor1]{Corresponding author}
\address[labelNAP]{Center for Applied Scientific Computing, LLNL,
  Livermore, CA 94550, USA.}
\address[labelFMG]{Department of Applied Mathematics, CU, Boulder, CO 80309, USA}
\address[labelAEC]{Department of Mathematics, SMU, Dallas, TX 75205, USA}
\address[labelYJR]{Department of Information Technology, UU, 751 05 Uppsala, SWEDEN}
\address[labelJLD]{Quantum Coherent Device Physics Group, LLNL, Livermore, CA
  94550, USA}


\begin{abstract}

  {\changethree This paper considers the optimal control problem for realizing
    logical gates in a closed quantum system. The quantum state is governed by Schr\"odinger's
    equation, which we formulate as a time-dependent Hamiltonian system in terms of the real and
    imaginary parts of the state vector. The system is discretized with the St\"ormer-Verlet scheme,
    which is a symplectic partitioned Runge-Kutta method.  Our main theoretical contribution is the
    derivation of a compatible time-discretization of the adjoint state equation, such that the
    gradient of the discrete objective function can be calculated exactly, at a computational cost
    of solving two Schr\"odinger systems, independently of the number of parameters in the control
    functions.
  
    A parameterization of the control functions based on B-splines with built-in carrier waves is
    also introduced. The carrier waves are used to specify the frequency spectra of the control
    functions, while the B-spline functions specify their envelope and phase. This approach allows
    the number of control parameters to be independent of, and significantly smaller than, the
    number of time steps for integrating Schr\"odinger's equation.

    We consider Hamiltonians that model the dynamics of a superconducting multi-level qudit and
    present numerical examples of how the proposed technique can be combined with the interior point
    L-BFGS algorithm from the IPOPT package for realizing quantum gates.
    In a set of test cases, the proposed algorithm is shown to compare favorably with
    QuTiP/pulse\_optim and Grape-Tensorflow.

}
\end{abstract}


\begin{keyword}
  Optimal control \sep Partitioned Runge-Kutta method \sep Discrete adjoint \sep Quantum computing



\end{keyword}

\end{frontmatter}


\section{Introduction}


A key challenge for realizing the potential of quantum computing lies in determining the most
efficient and accurate route to controlling the quantum states in a quantum device. This challenge
stems from the fact that current quantum computing systems, unlike classical computers, do not have
a fixed set of logical gates predetermined in hardware. Instead, the execution of a quantum
algorithm is carried out by first devising a set of classical control functions that are then
applied to the quantum computing hardware to guide the quantum states through a series of quantum
logical operations~\cite{Nielsen-Chuang}. Reducing the time required for a quantum gate to be
realized is critical for near-term quantum computing because it enables the computation to finish
before the quantum state collapses to a classical state, rendering the results meaningless. To
mitigate this problem, quantum optimal control techniques have been developed to produce customized
control pulses that minimize the execution time for complicated gates that directly map onto a
physical system~\cite{Shi_2019}.

{\changed Optimizing the control functions for realizing quantum gates is a optimal control problem
  where the objective function measures the infidelity of the gate transformation, constrained by
  Schr\"odinger's equation governing the evolution of the quantum states. For superconducting
  circuits it is also important to suppress leakage into highly energetic states~\cite{Leung-2017},
  leading to an optimal control problem in Mayer-Lagrange form.}  {\changethree Our approach builds
  upon the works of Hager~\cite{Hager2000}, Sanz-Serna~\cite{sanz2016symplectic} and
  Ober-Bl\"obaum~\cite{ober2008discrete}. Hager~\cite{Hager2000} first showed how the Hamiltonian
  structure in an optimization problem can be utilized to calculate the gradient of the objective
  function. Hager considered the case in which the state equation is discretized by one Runge-Kutta
  scheme, with the adjoint state equation discretized by another Runge-Kutta scheme. It was found
  that the discrete gradient can be calculated exactly if the pair of Runge-Kutta methods satisfy
  the requirements of a symplectic partitioned Runge-Kutta method. Further details and
  generalizations are described in the review paper by
  Sanz-Serna~\cite{sanz2016symplectic}. Ober-Bl\"obaum~\cite{ober2008discrete} extended Hager's
  approach to the case where the state equation itself is a Hamiltonian system that is discretized
  by a partitioned Runge-Kutta scheme. For autonomous state equations, it was shown that the
  compatible discretization of the adjoint state equation is another partitioned Runge-Kutta scheme.

In the quantum optimal control problem, the Schr\"odinger (state) equation is a time-dependent
Hamiltonian system. To ensure long-time numerical accuracy it is appropriate to discretize it using
a symplectic time-integration method~\cite{HairerLubichWanner-06}. For this purpose we use the
St\"ormer-Verlet method, which can be written as a partitioned Runge-Kutta scheme, based on the
trapezoidal and implicit midpoint rules. Our main theoretical contribution is the generalization of
Ober-Bl\"obaum's~\cite{ober2008discrete} work to the case of a time-dependent Hamiltonian system. We
show that the compatible method for the adjoint state equation resembles a partitioned Runge-Kutta
scheme, except that the time-dependent matrices must be evaluated at modified time levels.  }


Logical gates in a closed quantum system can be viewed as linear reversible mappings,
$|\bm{\psi''}\rangle = V_g | \bm{\psi'}\rangle$, from an initial state $|\bm{\psi'}\rangle$ to a
final state $|\bm{\psi''}\rangle$, where the reversibility implies that the mapping $V_g$ must be
unitary, $V_g^\dag V_g = I$.  To introduce the quantum control problem, we start by discussing the
case where the unitary transformation is defined in the entire $N$-dimensional state space, such
that it can be represented by a unitary matrix $V_g\in\mathbb{C}^{N\times N}$; a more general case
is described in Section~\ref{sec:generalized}.

In the following, we will replace the ket-notation~\cite{Nielsen-Chuang} of the state vector
$|\bm{\psi}\rangle = \psi^{(0)}|0\rangle + \psi^{(1)}|1\rangle + \ldots + \psi^{(N-1)}|N-1\rangle$
by the vector notation $\bm{\psi} = \psi^{(0)}\bm{e}_0 + \psi^{(1)} \bm{e}_1 + \ldots +
\psi^{(N-1)}\bm{e}_{N-1}$, which is more common in the computational mathematics
literature\footnote{Here, $\bm{e}_j$ represents the $j$th canonical basis vector, in which the $j$th
  element is one and all other elements are zero.}. The elements in the state vector are complex
probability amplitudes and the squared magnitude of the amplitudes sum to unity, i.e.,
$\|\bm{\psi}\|_2^2 = 1$.

To account for all admissible initial data in the Hilbert space $\mathbb{C}^{N}$, we consider the
evolutions from the canonical basis vectors $\bm{e}_j$, for $j=0,1,\ldots,N-1$. The time-dependent
control functions are expanded in terms of a finite number of basis functions, such that the control
functions are determined by the finite-dimensional parameter vector $\bm{\alpha}\in{\mathbb
  R}^D$. This leads to Schr\"odinger's equation in matrix form for the $N\times N$ complex-valued
solution operator matrix $U(t,\bm{\alpha})$:
\begin{equation}\label{eq:schrodinger_matrix}
\frac{dU}{dt} + i H(t,\bm{\alpha}) U = 0, \quad 0 \leq t \leq T, \quad U(0,\bm{\alpha}) = I_N,\quad H^\dag = H.
\end{equation}
Here, $I_N$ is the $N\times N$ identity matrix and $H(t,\bm{\alpha})$ is the Hamiltonian matrix, in
which the time-dependence is parameterized by $\bm{\alpha}$. As a result, the solution operator
matrix depends implicitly on $\bm{\alpha}$ through Schr\"odinger's equation. Due to linearity, the
solution for general initial conditions satisfies $\bm{\psi}(t,\bm{\alpha}) =
U(t,\bm{\alpha})\bm{\psi}(0,\bm{\alpha})$.

The goal of the quantum control problem is to determine the parameter vector $\bm{\alpha}$ such that
the time-dependence in the Hamiltonian matrix leads to a solution of Schr\"odinger's equation that
minimizes the difference between the target gate matrix $V_g$ and $U(T,\bm{\alpha})$. {\changed Here,
  we measure the difference by the commonly used target gate infidelity~\cite{qutip, Leung-2017,
    Lucarelli-2018, Machnes-2018, Shi_2019}, }
\begin{equation}\label{eq:objf}
  {\cal J}_0(U_T(\bm{\alpha})) = 1 - \frac{1}{N^2} \left| \mbox{Tr}\left(
   U^\dag_T(\bm{\alpha}) V_g\right) \right|^2,
      \quad U_T(\bm{\alpha}) := U(T,\bm{\alpha}).
\end{equation}
{\changed Because $U_T$ and $V_g$ are unitary, $|\mbox{Tr}(U_T^\dagger V_g)|\leq N$ and ${\cal J}_0
  \geq 0$. Note that the target gate infidelity is sensitive to relative phase differences between
  the columns of $U_T$ and $V_g$, but is invariant to global phase differences between $U_T$ and
  $V_g$.}



The quantum control problem is a constrained optimization problem where, in the basic setting, the
gate infidelity \eqref{eq:objf} is minimized under the constraints that the solution operator matrix
satisfies Schr\"odinger's equation \eqref{eq:schrodinger_matrix} and the amplitudes of the control
functions (determined by the parameter vector $\bm{\alpha}$) do not exceed prescribed limits. For a
discussion of the solvability of the quantum control problem, see for example Borzi et
al.~\cite{Borzi-17}.

While not a restriction of our approach, we exemplify our technique on Hamiltonians that
model the dynamics of a superconducting qudit (a qubit with more than two energy levels). We
represent the state vector in the energy basis in which the system Hamiltonian matrix is diagonal.
{\changed In the laboratory frame of reference, the Hamiltonian matrix is modeled by
\begin{equation}\label{eq_quantum-osc}
H_{lab}(t,\alphab) = \omega_a a^\dag a - \frac{\xi_a}{2} a^\dag a^\dag a a + f(t,\alphab)(a + a^\dagger).
\end{equation}
Here, $a$ and $a^\dagger$ are the lowering and raising matrices (see \ref{app_RotatingFrame}),
$\omega_a>0$ is the fundamental resonance frequency, $\xi_a>0$ is the self-Kerr coefficient and
$f(t,\bm{\alpha})$ is a real-valued control function that depends on the parameter vector
$\bm{\alpha}$.

To slow down the time scales in the state vector, the problem is transformed to a rotating frame of
reference (see~\ref{app_RotatingFrame}) in which the Hamiltonian matrix satisfies
\begin{align}
  H(t,\bm{\alpha}) = - \frac{\xi_a}{2} a^\dag a^\dag a a + p(t,\bm{\alpha}) ( a + a^\dag ) + i
  q(t,\bm{\alpha}) ( a - a^\dag ),\label{eq:hamiltonian-total}
\end{align}
where $p(t,\bm{\alpha})$ and $q(t,\bm{\alpha})$ are the real-valued control functions in the rotating frame of
reference. The control functions in the two frames are related by
\begin{align}\label{eq_rot-ansatz}
f(t) = 2p(t) \cos(\omega_a t) - 2q(t) \sin(\omega_a t).
\end{align}
}

Several numerical methods for the quantum control problem are based on the GRAPE
algorithm~\cite{Khaneja-2005}. In this case, Schr\"odinger's equation is discretized in time using
the second order accurate Magnus scheme~\cite{HairerLubichWanner-06}, in which the Hamiltonian
matrix is evaluated at the midpoint of each time step. A stair-step approximation of the control
functions is imposed such that each control function is constant within each time step. Thus, the
time step determines both the numerical accuracy of the dynamics of the quantum state {\em and} the
number of control parameters. With $Q$ control functions, $M$ time steps of size $h$, the control
functions are thus described by $M$ times $Q$ parameters $\alpha_{j,k}$. The propagator in the
Magnus method during the $j^{th}$ time step is of the form $\mbox{exp}(-ih(H_0 + \sum_k\alpha_{k,j}
H_k))$. In general, the matrices $H_0$ and $H_k$ do not commute, leading to an integral expression
for the derivative of the propagator with respect to the parameters, which is needed for computing
the gradient of the objective function. In the {\changed original} GRAPE method, this integral
expression is approximated by the first term in its Taylor series expansion, leading to an
approximate gradient that is polluted by an ${\cal O}(h^2)$ error. As the gradient becomes smaller
during the optimization, the approximation error will eventually dominate the numerical gradient,
which may hamper the convergence of the optimization algorithm. A more accurate way of numerically
evaluating the derivative of the time-step propagator can be obtained by retaining more terms in the
Taylor series expansion, or by using a matrix commutator expansion~\cite{de_Fouquieres_2011}.
{\changed More recently, the GRAPE algorithm has been generalized to optimize objective functions
  that include a combination of the target gate infidelity, integrals penalizing occupation of
  ``forbidden states'' (see Section~\ref{sec:generalized}) and terms for imposing smoothness and
  amplitude constraints on the control functions. Here, automatic differentiation is used for
  computing the gradient of the objective function~\cite{Leung-2017}. However, the number of control
  parameters is still proportional to the number of time steps, which may become very large when the
  duration of the gate is long, or the quantum state is highly oscillatory.}

{\changed As an alternative to calculating the gradient of the objective function by solving
  an adjoint equation backwards in time, the gradient can be calculated by differentiating
  Schr\"odinger's equation with respect to each parameter in the control function, leading to a
  differential equation for each component of the gradient of the state vector.} This approach,
implemented in the GOAT algorithm~\cite{Machnes-2018}, allows the gradient of the objective function
to be calculated exactly, but requires $(D+1)$ Schr\"odinger systems to be solved when the control
functions depend on $D$ parameters. This makes the method computationally expensive when the number
of parameters is large.

Using the stair-stepped approximation of the control functions often leads to a large number
{\changed of control} parameters, which may hamper the convergence of the GRAPE algorithm. The total
number of parameters can be reduced by instead expanding the control functions in terms of basis
functions.
By using the chain rule, the gradient from the GRAPE algorithm can then be used to calculate the
gradient with respect to the coefficients in the basis function expansion. This approach is
implemented in the GRAFS algorithm~\cite{Lucarelli-2018}, where the control functions are expanded
in terms of Slepian sequences.

Gradient-free optimization methods can also be applied to quantum optimal control
problems. These methods do not rely on the gradient to be evaluated and are therefore significantly
easier to implement. However, the convergence of these methods is usually much slower than for
gradient-based techniques, unless the number of control parameters is very small. One example of a
gradient-free methods for quantum optimal control is the CRAB algorithm~\cite{Caneva-2011}.


Many parameterizations of quantum control functions have been proposed in the literature, for
example cubic splines~\cite{Ewing-1990}, Gaussian pulse cascades~\cite{Emsley-1989}, Fourier
expansions~\cite{Zax-1988} and Slepian sequences~\cite{Lucarelli-2018}.
{\changed This paper presents a different approach, based on parameterizing the control functions by
  B-spline basis functions with carrier waves, see Figure~\ref{fig_bspline}. Our approach relies on
  the observation that transitions between the energy levels in a quantum system are triggered by
  resonance, at frequencies which often can be determined by inspection of the system
  Hamiltonian. The carrier waves are used to specify the frequency spectra of the control functions,
  while the B-spline functions specify their envelope and phase. We find that this approach allows
  the number of control parameters to be independent of, and significantly smaller than, the number of
  time steps for integrating Schr\"odinger's equation.}
\begin{figure}
    \centering
    \includegraphics[width=0.6\linewidth]{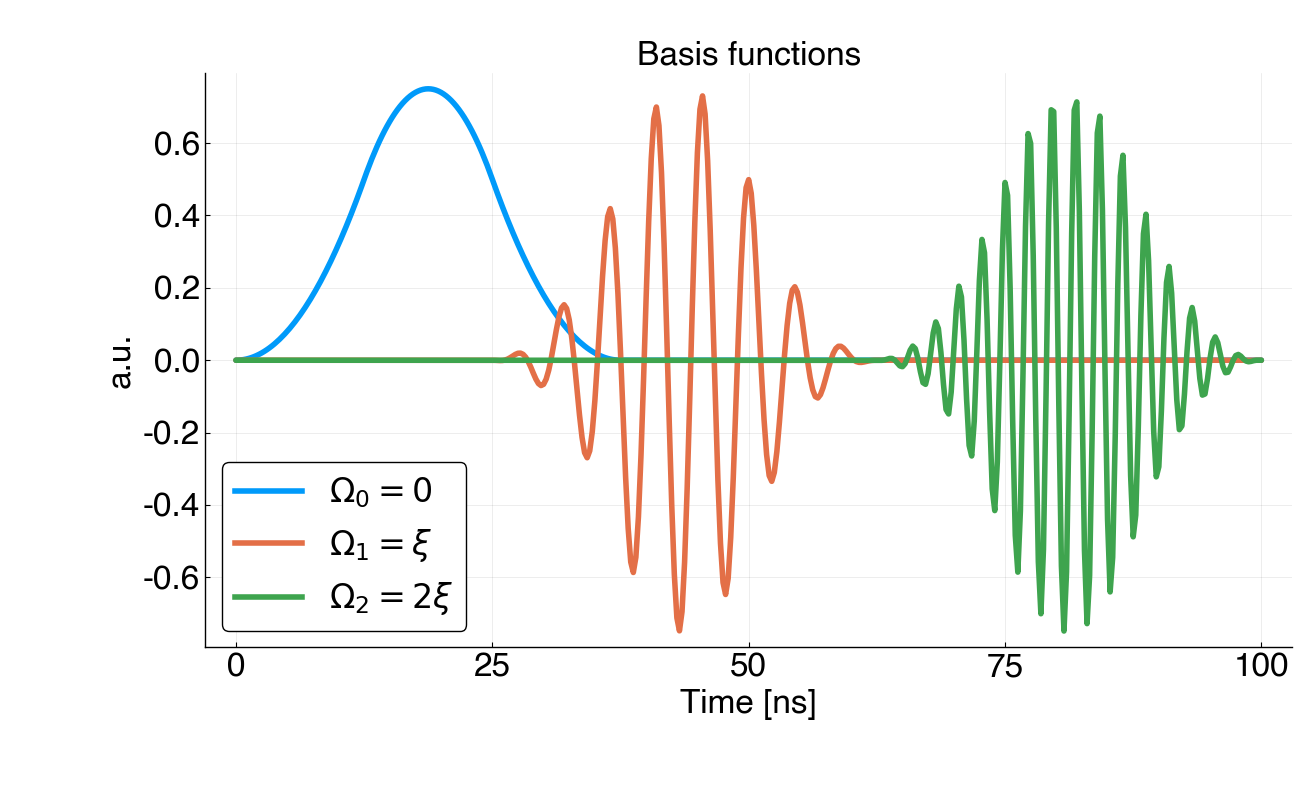}
    \caption{An example of three quadratic B-spline basis functions with carrier wave frequencies $(0,
      \xi, 2\xi)$.}  \label{fig_bspline} 
\end{figure}




The remainder of the paper is organized as follows. In Section~\ref{sec:generalized}, we generalize
the optimization problem to the case of target gates that are defined in a subspace of the entire
state space. In Section~\ref{sec_real}, we first introduce the real-valued formulation of
Schr\"odinger's equation, followed by a presentation of the symplectic St\"ormer-Verlet
time-stepping method, written as a partitioned Runge-Kutta scheme.
To achieve an exact gradient of the discrete objective function, in Section~\ref{sec_disc} we derive
the discrete adjoint time integration method.  {\changed This method resembles a partitioned
  Runge-Kutta scheme, except that the time-dependent matrices are evaluated at modified
  time-levels.}  The solution of the discrete adjoint equation is used to efficiently calculate all
components of the gradient of the discrete objective function. The parameterization of the control
functions using B-splines with carrier waves is presented in
Section~\ref{sec:B-Splines}. Section~\ref{sec_numopt} presents a numerical
example of how the proposed technique can be combined with the interior point L-BFGS
algorithm~\cite{Nocedal-Wright} from the IPOPT package~\cite{Wachter2006} to realize multi-level
qudit gates. Important properties of the optima are exposed by analyzing the eigenvalues of the
Hessian.  {\changed The proposed algorithm has been implemented in the JuQBox package, written in
  the Julia~\cite{julia} programming language. In Section~\ref{sec_compare}, we compare its
  performance to two variants of the GRAPE algorithm.} Concluding remarks are given in
Section~\ref{sec_conc}.



\section{Generalized gates} \label{sec:generalized}

In quantum computing applications it is common to define gate transformations in a subspace of the
entire (possibly infinite dimensional) state space, in which the evolution of higher energy states
is not relevant for the gate transformation, but if left uncontrolled, may lead to leakage of
probability.
In the following, let the subspace of interest contain $E>0$ ``essential'' states and let $G = N - E
\geq 0$ denote the number of ``guard'' states. The guard states that correspond to the highest
energy levels in the model are often called ``forbidden'' states~\cite{Leung-2017}.

In the case of one qudit oscillator, we can always order the elements in the state vector such that
they correspond to increasing energy levels.
The Schr\"odinger equation governs the evolution of all energy levels in the state vector, including
the guard levels, but the unitary gate transformation is only defined in the subspace of the
essential states. {\changed This requirement leads us to define the target gate transformation
  matrix according to
\begin{align}\label{eq_target-partitioned}
    V = \begin{bmatrix}
      V_g \\
      \bm{0}
    \end{bmatrix} \in \mathbb{C}^{N\times E},\quad V_g \in \mathbb{C}^{E\times E},\quad V_g^\dagger V_g =I_E.
\end{align}
}

Let the state vector $\bm{\psi}_j(t,\bm{\alpha})\in {\mathbb C}^N$ satisfy the Schr\"odinger equation,
\begin{equation}\label{eq:schrodinger_example}
  \frac{d \bm{\psi}_j}{dt} + iH(t,\bm{\alpha})\bm{\psi}_j = 0,\quad 0\leq t\leq T, \quad
  \bm{\psi}_j (0,\bm{\alpha}) = \bm{e}_{j},
\end{equation}
for $j=0,1,\ldots,E-1$.
The solution operator matrix $U(t,\bm{\alpha})$ and the target gate matrix $V$ are rectangular with $N$ rows and $E$ columns,
\begin{equation}\label{eq_solution-matrix}
U(t,\bm{\alpha}) = [ \bm{\psi}_{0}(t,\bm{\alpha}), \bm{\psi}_{1}(t,\bm{\alpha}), \ldots, \bm{\psi}_{E-1}(t,\bm{\alpha}) ],\quad
V = [ \bm{d}_{0}, \bm{d}_{1}, \ldots , \bm{d}_{E-1}].
\end{equation}
{\changed The decomposition \eqref{eq_target-partitioned} implies that the last $G$ rows of $\bm{d}_j$ must be zero. }
%

The matrix overlap function $R_V(U_T)$ in \eqref{eq:objf} generalizes in a straightforward way to unitary gates that are defined in the subspace, resulting in the target gate infidelity function
\begin{equation}\label{eq:objf-essential}
  {\cal J}_1(U_T(\bm{\alpha})) = 1 - \frac{1}{E^2} \left| S_V(U_T(\bm{\alpha})) \right|^2, \quad
  S_V(U_T(\bm{\alpha})) = \sum_{j=0}^{E-1} \left\langle \bm{\psi}_j(T,\bm{\alpha}), \bm{d}_j \right\rangle_2,
\end{equation}
where $\langle \cdot, \cdot \rangle_2$ is the $\ell_2$ vector scalar product. The population of the
guard states can be measured by the objective function
\begin{equation}\label{eq:objf-guard}
  {\cal J}_2(U(\cdot,\bm{\alpha})) = \frac{1}{T} \int_0^T \sum_{j=0}^{E-1}
  \left\langle \bm{\psi}_j(t,\bm{\alpha}), W\bm{\psi}_j(t,\bm{\alpha}) \right\rangle_2 \, dt.
\end{equation}
Here, $W$ is a diagonal $N\times N$ positive semi-definite weight matrix. The elements in $W$ are
zero for all essential states and are positive for the guard states. The elements of $W$ are typically
larger for higher energy levels in the model.

For the quantum control problem with guard states, we formulate the optimization problem as
\begin{gather}
  \mbox{min}_{\bm{\alpha}}\,  {\cal G}(\bm{\alpha}) := {\cal
    J}_1(U_T(\bm{\alpha})) + {\cal J}_2(U(\cdot,\bm{\alpha})), \label{eq:objf-total}\\ 
  \frac{dU}{dt} + iH(t, \bm{\alpha}) U = 0,\quad 0\leq t\leq T, \quad
  U(0,\bm{\alpha}) = [\bm{e}_{0}, \bm{e}_{1}, \ldots,
    \bm{e}_{E-1}].\label{eq:schrodinger-matrix-essential} \\
  \alpha_{min} \leq \alpha_q \leq \alpha_{max},\quad q=1,2,\ldots,D.\label{eq_ineq-constraints2}
\end{gather}
In the special case of zero guard states, ${\cal J}_2(U)=0$ because
$W=0$. Thus, the above formulation applies to both the cases with and without guard states,
i.e., when $G=N-E\geq 0$.

\section{Real-valued formulation}\label{sec_real}

%
A real-valued formulation of Schr\"odinger's equation \eqref{eq:schrodinger_example} is given by
\begin{equation}\label{eq_real-schrodinger}
  \begin{bmatrix}
    \dot{\bm{u}}\\ \dot{\bm{v}}
  \end{bmatrix} =
  \begin{bmatrix}
    S(t) & -K(t) \\ K(t) & S(t)
  \end{bmatrix}
  \begin{bmatrix} \bm{u}\\ \bm{v} \end{bmatrix} =:
  %
%
  \begin{bmatrix}
    f^u(\bm{u},\bm{v},t)\\
    f^v(\bm{u},\bm{v},t)
  \end{bmatrix},\quad
  \begin{bmatrix}
    \bm{u}(0)\\
    \bm{v}(0)
  \end{bmatrix}
  =
  \begin{bmatrix}
    \bm{g}^u\\
    \bm{g}^v
  \end{bmatrix},
\end{equation}
where,
\[
  \bm{u} = \mbox{Re}(\bm{\psi}),\quad \bm{v} = -\mbox{Im}(\bm{\psi}),\quad 
  K = \mbox{Re}\,(H),\quad S = \mbox{Im}\,(H),
\]
%
Because the matrix $H$ is Hermitian, $K^T=K$ and $S^T=-S$ (note that the matrix $S$ is unrelated to
the matrix overlap function $S_V$). The real-valued formulation of Schr\"odinger's equation is a
time-dependent Hamiltonian system corresponding to the Hamiltonian functional,
\begin{equation}\label{eq_hamiltonian}
{\cal H}(\bm{u},\bm{v},t) = \bm{u}^T S(t) \bm{v} + \frac{1}{2} \bm{u}^T K(t)\bm{u} + \frac{1}{2}
\bm{v}^T K(t) \bm{v}.
\end{equation}
In general, $S(t)\ne 0$, which makes the Hamiltonian system non-separable.

In terms of the real-valued formulation, the columns of the solution operator matrix in
\eqref{eq_solution-matrix} satisfy $U = \left[ \bm{u}_1 -i\bm{v}_1,\ \bm{u}_2 -i\bm{v}_2,\ \ldots,
  \bm{u}_E -i\bm{v}_E \right]$. Here, $(\bm{u}_j, \bm{v}_j)$ satisfy \eqref{eq_real-schrodinger}
subject to the initial conditions $\bm{g}^u_j = \bm{e}_{j}$ and $\bm{g}^v_j = \bm{0}$.  The columns
in the target gate matrix $V$ correspond to
\[
V = \left[ \bm{d}_1^u - i
  \bm{d}_1^v,\  \bm{d}_2^u - i \bm{d}_2^v,\ \ldots,  \bm{d}_E^u - i \bm{d}_E^v\right],\quad 
\bm{d}_j^u = \mbox{Re}(\bm{d}_j),\quad \bm{d}_j^v = -\mbox{Im}(\bm{d}_j).
\]
Using the real-valued notation, the objective function \eqref{eq:objf-total} can be written
\begin{multline}\label{eq_objf-total-real}
  {\cal G}(\bm{\alpha}) =
  \left(1 - \frac{1}{E^2} \left| S_V( U_T(\bm{\alpha}) ) \right| ^2
  \right) \\
  + \frac{1}{T} \sum_{j=0}^{E-1} \int_{0}^T \Bigl( \langle \bm{u}_j(t,\bm{\alpha}), W \bm{u}_j(t,\bm{\alpha}) \rangle_2 +
  \langle \bm{v}_j(t,\bm{\alpha}), W \bm{v}_j(t,\bm{\alpha}) \rangle_2\Bigr) \, dt,
\end{multline}
where,
\begin{multline}\label{eq_st-real}
  S_V( U_T ) =
  \sum_{j=0}^{E-1}\left(  \left\langle \bm{u}_j(T,\bm{\alpha}), \bm{d}^u_j \right\rangle_2
  +  \left\langle \bm{v}_j(T,\bm{\alpha}), \bm{d}^v_j \right\rangle_2 \right)\\
  + i\sum_{j=0}^{E-1} \left( \left\langle \bm{v}_j(T,\bm{\alpha}), \bm{d}^u_j \right\rangle_2
  - \left\langle \bm{u}_j(T,\bm{\alpha}), \bm{d}^v_j \right\rangle_2 \right).
\end{multline}

\subsection{Time integration}\label{sec_time-stepping}

Let $t_n= nh$, for $n=0,1,\ldots,M$, be a uniform grid in time where $h=T/M$ is the time step. Also
let $\bm{u}^n\approx \bm{u}(t_n)$ and $\bm{v}^n\approx \bm{v}(t_n)$ denote the numerical solution on
the grid.
We use a partitioned Runge-Kutta (PRK) scheme~\cite{HairerLubichWanner-06}
to discretize the real-valued formulation of Schr\"odinger's equation,
\begin{alignat}{3}
  \bm{u}^0 &= \bm{g}^u,\quad &
  \bm{v}^0 &= \bm{g}^v,\\
  \bm{u}^{n+1} &= \bm{u}^n + h \sum_{i=1}^s b^u_i \bm{\kappa}^{n,i},\quad &
  \bm{v}^{n+1} &= \bm{v}^n + h \sum_{i=1}^s b^v_i \bm{\ell}^{n,i},\label{eq_uv-update}\\
  \bm{\kappa}^{n,i} &= f^u(\bm{U}^{n,i}, \bm{V}^{n,i}, t_n + c^u_i h),\quad &
  \bm{\ell}^{n,i} &= f^v(\bm{U}^{n,i}, \bm{V}^{n,i}, t_n + c^v_i h),\label{eq_uv-slopes}\\
  \bm{U}^{n,i} &= \bm{u}^n + h \sum_{j=1}^s a^u_{ij} \bm{\kappa}^{n,j},\quad &
  \bm{V}^{n,i} &= \bm{v}^n + h \sum_{j=1}^s a^v_{ij}\bm{\ell}^{n,j}. \label{eq_stage-values}
\end{alignat}
Here, $s\geq 1$ is the number of stages. The stage variables $\bm{U}^{n,i}$ and
$\bm{V}^{n,i}$ are set in a bold font to indicate that they are unrelated to the solution operator
matrix $U(t,\bm{\alpha})$ and the target gate matrix $V$.

The St\"ormer-Verlet scheme is a two-stage PRK method ($s=2$) that is symplectic, time-reversible and
second order accurate~\cite{HairerLubichWanner-06}. It combines the trapezoidal and the implicit midpoint rules, with Butcher coefficients:
%
%
\begin{align}
    a^u_{11} &= a^u_{12} = 0,\quad a^u_{21} = a^u_{22} = \frac{1}{2},\qquad&   
    a^v_{11} &= a^v_{21} = \frac{1}{2},\quad a^v_{12} = a^v_{22} = 0,\\
    b^u_{1} &= b^u_2 = \frac{1}{2},\quad  c^u_1 = 0,\quad c^u_2=1,\qquad&   
    b^v_{1} &= b^v_2 = \frac{1}{2},\quad  c^v_1 = c^v_2=\frac{1}{2}.
\end{align}


\subsection{Time step restrictions for accuracy and stability}\label{sec_time-step-est}

The accuracy in the numerical solution of Schr\"odinger's equation is essentially determined by how well the fastest time scale in the state vector is resolved on the grid in time. The analysis of the time scales in the solution of Schr\"odinger's equation is most straightforward to perform in the complex-valued formulation \eqref{eq:schrodinger_example}. 

{\changed There are two fundamental time scales that must be resolved in the solution of Schr\"odinger's equation. The first corresponds to how quickly the control functions must vary in time to trigger the desired transitions between the energy levels in the quantum system. This time scale is determined by the transition frequencies in the system Hamiltonian, which follow as the difference between its consecutive eigenvalues. In the Hamiltonian model~\eqref{eq:hamiltonian-total}, the angular transition frequencies between the essential energy levels are
\begin{align}
    \Delta_{j} = j\xi_a,\quad j=0,\ldots,E-1.
\end{align}

The second time scale is due to the harmonic oscillation of the phase in the state vector. It} can be estimated by freezing the time-dependent coefficients in the Hamiltonian matrix at some time $t=t_*$ and considering Schr\"odinger's equation with the time-independent Hamiltonian matrix $H_*=H(t_*)$. The $N\times N$ matrix $H_*$ is Hermitian and can be diagonalized by a unitary transformation,
\[
H_* X = X \Gamma,\quad X^\dag X = I_N,\quad \Gamma = \mbox{diag$(\gamma_1, \gamma_2, \ldots, \gamma_N)$},
\]
where the eigenvalues $\gamma_k$ are real. By the change of variables $\widetilde{\bm{\psi}} = X^\dag \bm{\psi}$, the solution of the diagonalized system follows as
\[
\widetilde{\psi}_k(t) = e^{-i\gamma_k t} \widetilde{\psi}_k(0),
\]
corresponding to the period $\tau_k = 2\pi/|\gamma_k|$. {\changed The shortest period thus follows
  from the spectral radius of $H_*$, $\rho(H_*) = \max_k |\gamma_k|$.}

{\changed To estimate the time step for the St\"ormer-Verlet method, we require that} the shortest
period in the solution of Schr\"odinger's equation must be resolved by at least $C_P$ time
steps. {\changed Taking both time scales into account} leads to the time step restriction
\begin{align}\label{eq_timestep}
h \leq 
\frac{2\pi}{C_P \max\{\rho(H_*), \max_{j}(|\Delta_j|)\}}.
\end{align}
The value of $C_P$ that is needed to obtain a given accuracy in the numerical solution depends on
the order of accuracy, the duration of the time integration, as well as the details of the
time-stepping scheme. For second order accurate methods such as the St\"ormer-Verlet method,
acceptable accuracy for engineering applications can often achieved with $C_P\approx 40$. With the
St\"ormer-Verlet method, we note that the time-stepping can become unstable if $C_P\leq 2$,
corresponding to a sampling rate below the Nyquist limit.

{\changed After freezing the coefficients, the Hamiltonian~\eqref{eq:hamiltonian-total} becomes
\begin{align*}
H_* = -\frac{\xi_a}{2}a^\dagger a^\dagger a a + p_*(a + a^\dag) + i q_* (a - a^\dag),\quad 
p_* = p(t_*,\bm{\alpha}),\quad 
q_* = q(t_*,\bm{\alpha}).
\end{align*}
We can estimate the spectral radius of $H_*\in\mathbb{C}^{N\times N}$ using the Gershgorin circle
theorem~\cite{Golub-VanLoan}.
Because $H_*$ is Hermitian, all its eigenvalues are real. As a result, its spectral radius can be
bounded by
\begin{align*}
    \rho(H_*)
    \leq \frac{|\xi_a|}{2}(N-1)(N-2) + (|p_*| + |q_*|)\sqrt{N-1}.
\end{align*}
Hence, it is the largest value of $(|p_*| + |q_*|)$ that determines the time step. 

Given the parameter vector $\bm{\alpha}$, the control functions are bounded by $p_\infty = \max_t |p(t,\alphab)|$ and $q_\infty = \max_t |q(t,\alphab)|$, where the maximum is evaluated for times $0\leq t \leq T$. Thus, using the estimate
\begin{align}\label{eq_spec-rad}
    \rho(H_*) 
    \leq \frac{|\xi_a|}{2}(N-1)(N-2) + (p_\infty + q_\infty)\sqrt{N-1},
\end{align}
in \eqref{eq_timestep} guarantees that the time-dependent phase in the state vector is resolved by at at least $C_P$ time steps per shortest period.
}

If the optimization imposes amplitude constraints on the parameter vector, $|\bm{\alpha}|_\infty \leq \alpha_{max}$, those constraints can be used to estimate the time step before the optimization starts. This allows the same time step to be used throughout the iteration and eliminates the need to recalculate the spectral radius of $H_*$ when $\bm{\alpha}$ changes.

{\changed Our implementation of the St\"ormer-Verlet scheme was verified to be second order accurate. It was also found to give approximately the same accuracy as the second order Magnus integrator~\cite{HairerLubichWanner-06} when the same time step was used in both methods (data not shown to conserve space).}

\section{Discretizing the objective function and its gradient}\label{sec_disc}

In this section, we develop a ``discretize before optimize'' approach in which we first discretize the objective function and then derive a compatible scheme for discretizing the adjoint state equation, which is used for computing the gradient of the objective function. 
{\changethree 
As was outlined in the introduction, our approach builds upon the works of Hager~\cite{Hager2000}, Sanz-Serna~\cite{sanz2016symplectic} and Ober-Bl\"obaum~\cite{ober2008discrete}.}

\subsection{Discretizing the objective function}\label{sec_discrete-obj}

The St\"ormer-Verlet scheme can be written in terms of the stage variables
$(\bm{U}^{n,i},\bm{V}^{n,i})$ by substituting $(\bm{\kappa}^{n,i}, \bm{\ell}^{n,i})$ from
\eqref{eq_uv-slopes} into \eqref{eq_uv-update},
\begin{align}
  \bm{u}^0 &= \bm{g}^u ,\quad \bm{v}^0 = \bm{g}^v,\label{eq_initial-cond}\\ \bm{u}^{n+1} &= \bm{u}^n
  + \frac{h}{2}\left(S_n\bm{U}^{n,1} + S_{n+1}\bm{U}^{n,2} - K_n\bm{V}^{n,1} - K_{n+1}\bm{V}^{n,2}
  \right),\label{eq_u-update}\\
  \bm{v}^{n+1} &= \bm{v}^n + \frac{h}{2}\left(K_{n+1/2} \left(\bm{U}^{n,1} + \bm{U}^{n,2}\right)
  + S_{n+1/2} (\bm{V}^{n,1} + \bm{V}^{n,2}) \right),\label{eq_v-update}
\end{align}
and into \eqref{eq_stage-values},
\begin{align} 
  \bm{U}^{n,1} &= \bm{u}^n,\label{eq_U1}\\ \bm{U}^{n,2} &= \bm{u}^n +
  \frac{h}{2}\left(S_n\bm{U}^{n,1} + S_{n+1}\bm{U}^{n,2} - K_n\bm{V}^{n,1} - K_{n+1}\bm{V}^{n,2}
  \right),\label{eq_U2}\\ \bm{V}^{n,1} &= \bm{v}^n + \frac{h}{2}\left( K_{n+1/2} \bm{U}^{n,1} +
  S_{n+1/2} \bm{V}^{n,1} \right),\label{eq_V1}\\ \bm{V}^{n,2} &= \bm{v}^n + \frac{h}{2}\left(
  K_{n+1/2} \bm{U}^{n,1} + S_{n+1/2} \bm{V}^{n,1} \right).\label{eq_V2}
\end{align}
Here, $S_n = S(t_{n})$, $S_{n+1/2} = S(t_n + 0.5 h)$, etc. Because $S(t)\ne0$, the
scheme is block implicit.
Note that $\bm{u}^{n+1} = \bm{U}^{n,2}$ and $\bm{V}^{n,1} = \bm{V}^{n,2} = \bm{v}(t_{n+1/2}) + {\cal O}(h^2)$.

The numerical solution at the final time step provides a second order accurate approximation of the
continuous solution operator matrix $U_T$, which we denote $U_{Th}$. It is used to approximate the
matrix overlap function $S_V(U_T)$ in \eqref{eq_st-real},
\begin{equation}
  S_{Vh}(U_{Th}) =
  \sum_{j=0}^{E-1}\left( \left\langle \bm{u}_j^M, \bm{d}^u_j \right\rangle_2
  +  \left\langle \bm{v}_j^M, \bm{d}^v_j \right\rangle_2 \right) 
  + i\sum_{j=0}^{E-1} \left( \left\langle \bm{v}_j^M, \bm{d}^u_j \right\rangle_2
  - \left\langle \bm{u}_j^M, \bm{d}^v_j \right\rangle_2 \right),
\end{equation}
which is then used as the first part of the discrete objective function,
\begin{equation}\label{eq:cost_infidelity}
  {\cal J}_{1h}(U_{Th}) = \left(1 - \frac{1}{E^2} | S_{Vh}(U_{Th}) |^2 \right).
\end{equation}

The integral in the objective function \eqref{eq_objf-total-real} can be discretized to second order
accuracy by using the Runge-Kutta stage variables,
\begin{multline}\label{eq:cost_guardLevel}
  {\cal J}_{2h}(\bm{U}, \bm{V}) = \frac{h}{T}\sum_{j=0}^{E-1}
    \sum_{n=0}^{M-1}
    \left(  \frac{1}{2}\left \langle \bm{U}^{n,1}_j, W \bm{U}^{n,1}_j \right \rangle_2
    +  \frac{1}{2} \left \langle \bm{U}^{n,2}_j, W \bm{U}^{n,2}_j \right \rangle_2 \right.\\
    + \left. \left \langle \bm{V}^{n,1}_j, W \bm{V}^{n,1}_j \right \rangle_2
    \right).
\end{multline}
Based on the above formulas we discretize the objective function \eqref{eq_objf-total-real}
according to
\begin{equation}\label{eq:FullyDiscreteObjective}
  {\cal G}_h(\bm{\alpha}) = {\cal J}_{h}(U_{Th}^\alpha, \bm{U}^\alpha,
  \bm{V}^\alpha),\quad
  {\cal J}_h(U_{Th},\bm{U},\bm{V}) := {\cal J}_{1h}(U_{Th}) + {\cal J}_{2h}(\bm{U}, \bm{V}).
\end{equation}
Here, $U_{Th}^\alpha$, $\bm{U}^\alpha$ and $\bm{V}^\alpha$ represent the
time-discrete solution of the St\"ormer-Verlet scheme for a given parameter vector $\bm{\alpha}$.
We note that ${\cal G}_h(\bm{\alpha})$ can be evaluated by accumulation during the time-stepping of
the St\"ormer-Verlet scheme.

\subsection{The discrete adjoint approach}\label{sec_discrete-adjoint}


The gradient of the discretized objective function can be derived from first order optimality
conditions of the corresponding discrete Lagrangian.
In this approach, let $(\bm{\mu}_j^n,\bm{\nu}_j^n)$ be the adjoint variables and let
$(\bm{M}_j^{n,i},\bm{N}_j^{n,i})$ be Lagrange multipliers. We define the discrete Lagrangian by
\begin{multline}\label{eq_disc-Lagrangian}
  {\cal L}_h(\bm{u},\bm{v},\bm{U},\bm{V}, \bm{\mu},\bm{\nu},\bm{M},\bm{N},\bm{\alpha}) = \\
  {\cal J}_{h}(U_{Th}, \bm{U}, \bm{V}) 
  - \sum_{j=0}^{E-1} \left( \left\langle \bm{u}_j^0 - \bm{g}_j^u, \bm{\mu}_j^0\right\rangle_2 + \left\langle \bm{v}_j^0
  - \bm{g}_j^v, \bm{\nu}_j^0\right\rangle_2 + \sum_{k=1}^6 T_j^k \right).
\end{multline}
The first two terms in the sum enforce the initial conditions \eqref{eq_initial-cond}. The terms
$T^1_j$ and $T^2_j$ enforce the time-stepping update formulas
\eqref{eq_u-update}-\eqref{eq_v-update} in the St\"ormer-Verlet scheme,
\begin{align}
T^1_j &= \sum_{n=0}^{M-1} \left\langle \bm{u}_j^{n+1} - \bm{u}_j^n - \frac{h}{2}\left(S_n\bm{U}_j^{n,1} +
S_{n+1}\bm{U}_j^{n,2} - K_n\bm{V}_j^{n,1} - K_{n+1}\bm{V}_j^{n,2}  \right), \bm{\mu}_j^{n+1} \right\rangle_2,\\
T^2_j &= \sum_{n=0}^{M-1} \left\langle \bm{v}_j^{n+1} - \bm{v}_j^n - \frac{h}{2}\left(K_{n+1/2}
\left(\bm{U}_j^{n,1} + \bm{U}_j^{n,2}\right) + S_{n+1/2} (\bm{V}_j^{n,1} + \bm{V}_j^{n,2}) \right), \bm{\nu}_j^{n+1} \right\rangle_2.
\end{align}
{\changethree The terms $T_j^3$ to $T_j^6$ enforce the relations between the stage variables
  \eqref{eq_U1}-\eqref{eq_V2} using the Lagrange multipliers $(\bm{M}_j^{n,i}$ and
  $\bm{N}_j^{n,i})$, see~\ref{appendix:adjoint} for details.}

To derive the discrete adjoint scheme, we note that the discrete Lagrangian \eqref{eq_disc-Lagrangian}
has a saddle point if
\begin{align}
  \frac{\p{\cal L}_h}{\p \bm{\mu}_j^n} =
  \frac{\p{\cal L}_h}{\p \bm{\nu}_j^n} =
  \frac{\p{\cal L}_h}{\p \bm{N}_j^{n,i}} = 
  \frac{\p{\cal L}_h}{\p \bm{M}_j^{n,i}} & = 0 ,\label{eq_forward}\\
  \frac{\p{\cal L}_h}{\p \bm{u}_j^n} =
  \frac{\p{\cal L}_h}{\p \bm{v}_j^n} =
  \frac{\p{\cal L}_h}{\p \bm{U}_j^{n,i}} =
  \frac{\p{\cal L}_h}{\p \bm{V}_j^{n,i}} &= 0,\label{eq_backward}
\end{align}
for $n=0,1,\ldots,M$, $i=1,2$ and $j=0,1,\ldots,E-1$. Here, the set of conditions in
\eqref{eq_forward} result in the St\"ormer-Verlet scheme \eqref{eq_initial-cond}-\eqref{eq_V2} for
evolving $(\bm{u}_j^n, \bm{v}_j^n, \bm{U}_j^{n,i}, \bm{V}_j^{n,i})$ forwards in time. The set of
conditions in \eqref{eq_backward} result in a time-stepping scheme for evolving the adjoint
variables $(\bm{\mu}_j^n, \bm{\nu}_j^n)$ backwards in time, as is made precise in the following
lemma.
\begin{lemma}\label{lem_adjoint-bck}
  Let ${\cal L}_h$ be the discrete Lagrangian defined by \eqref{eq_disc-Lagrangian}. Furthermore, let
  $(\bm{u}_j^n, \bm{v}_j^n, \bm{U}_j^{n,i}, \bm{V}_j^{n,i})$ satisfy the St\"ormer-Verlet scheme
  \eqref{eq_initial-cond}-\eqref{eq_V2} for a given parameter vector $\bm{\alpha}$. Then, the set of
  saddle-point conditions \eqref{eq_backward} are satisfied if the Lagrange multipliers
  $(\bm{\mu}_j^n, \bm{\nu}_j^n)$ are calculated according to the reversed time-stepping scheme,
\begin{align}
  \bm{\mu}_j^{M} & = \pdiff{{\cal J}_h}{\bm{u}_j^M}, \quad
  \bm{\nu}_j^{M} = \pdiff{{\cal J}_h}{\bm{v}_j^M}, \\
  \bm{\mu}_j^n &= \bm{\mu}_j^{n+1} - \frac{h}{2} \left(\bm{\kappa}_j^{n,1} +
  \bm{\kappa}_j^{n,2}\right),\label{eq_mu}\\ 
  \bm{\nu}_j^{n} &= \bm{\nu}_j^{n+1} - \frac{h}{2} \left( \bm{\ell}_j^{n,1} + \bm{\ell}_j^{n,2} \right) ,\label{eq_nu}
\end{align}
for $n=M -1, M -2, \ldots 0$. Because $S^T = -S$ and $K^T = K$, the slopes satisfy
\begin{align}
\bm{\kappa}_j^{n,1} &= S_n \bm{X}_j^{n} - K_{n+1/2} \bm{Y}_j^{n,1} - \frac{2}{h} \frac{\p {\cal
    J}_h}{\p \bm{U}_j^{n,1}} ,\label{eq_kappa1}\\
\bm{\kappa}_j^{n,2} &= S_{n+1} \bm{X}_j^{n} -
K_{n+1/2} \bm{Y}_j^{n,2} - \frac{2}{h} \frac{\p {\cal J}_h}{\p
  \bm{U}_j^{n,2}},\label{eq_kappa2}\\
\bm{\ell}_j^{n,1} &= K_{n} \bm{X}_j^{n} + S_{n+1/2}
\bm{Y}_j^{n,1} - \frac{2}{h} \frac{\p {\cal J}_h}{\p
  \bm{V}_j^{n,1}},\label{eq_ell1}\\
\bm{\ell}_j^{n,2} &= K_{n+1} \bm{X}_j^{n} + S_{n+1/2}
\bm{Y}_j^{n,2} - \frac{2}{h} \frac{\p {\cal J}_h}{\p \bm{V}_j^{n,2}},\label{eq_ell2}
\end{align}
where the stage variables are given by
\begin{align}
  \bm{X}_j^{n}   &= \bm{\mu}_j^{n+1} - \frac{h}{2} \bm{\kappa}_j^{n,2}, \label{eq_adjoint-stageX}\\
  \bm{Y}_j^{n,2} &= \bm{\nu}_j^{n+1},\label{eq_Y1}\\
  \bm{Y}_j^{n,1} &= \bm{\nu}_j^{n+1} -\frac{h}{2} \left( \bm{\ell}_j^{n,1} + \bm{\ell}_j^{n,2} \right).\label{eq_Y2}
\end{align}
\end{lemma}
\begin{proof}
The lemma follows after a somewhat tedious but straightforward calculation shown in detail
in~\ref{appendix:adjoint}.
\end{proof}



{\changethree Corresponding to the continuous Schr\"odinger equation~\eqref{eq_real-schrodinger},
  the adjoint state equation (without forcing) is
\begin{equation}\label{eq_real-adjoint}
  \begin{bmatrix}
    \dot{\bm{\mu}}\\ \dot{\bm{\nu}}
  \end{bmatrix} =
  \begin{bmatrix}
    S(t) & -K(t) \\ 
    K(t) & S(t)
  \end{bmatrix}
  \begin{bmatrix} 
  \bm{\mu}\\ 
  \bm{\nu} 
  \end{bmatrix} 
  %
  =:
  \begin{bmatrix}
    \bm{f}^\mu(\bm{\mu},\bm{\nu},t)\\
    \bm{f}^\nu(\bm{\mu},\bm{\nu},t)
  \end{bmatrix},
\end{equation}
where we used that $S^T = - S$ and $K^T = K$.
\begin{corollary}
The time-stepping scheme \eqref{eq_mu}-\eqref{eq_Y2} (without forcing) is a consistent approximation
of the continuous adjoint state equation \eqref{eq_real-adjoint}. It can be written as a modified
partitioned Runge-Kutta method, where the Butcher coefficients are
\begin{align}
    a^\mu_{11} &= a^\mu_{21} = 1/2,\quad a^\mu_{12} = a^\mu_{22} = 0,\quad& 
    a^\nu_{11} &= a^\nu_{12} = 0, \quad  a^\nu_{21} = a^\nu_{22} = 1/2,\\
    b^\mu_{1} &= b^\mu_2 = \frac{1}{2},\quad  &   
    b^\nu_{1} &= b^\nu_2 = \frac{1}{2},
\end{align}
corresponding to the implicit midpoint rule for the $\bm{\mu}$-equation and the trapezoidal rule for
the $\bm{\nu}$-equation in \eqref{eq_real-adjoint}. The modifications to the partitioned Runge-Kutta
scheme concerns the formulae for the slopes, \eqref{eq_kappa1}-\eqref{eq_ell2}.
Because of the time-levels at which the matrices $K$ and $S$ are evaluated, it is {\em not} possible to define
Butcher coefficients $c^\mu_i$ and $c^\nu_i$ such that 
\begin{align*}
    \bm{\kappa}_j^{n,i} &= \bm{f}^\mu(\bm{X}_j^{n,i}, \bm{Y}_j^{n,i}, t_n + c^\mu_ih),\\
    \bm{\ell}_j^{n,i} &= \bm{f}^\nu(\bm{X}_j^{n,i}, \bm{Y}_j^{n,i}, t_n + c^\nu_ih).
\end{align*}
\end{corollary}
\begin{proof} See~\ref{app_cor1}.
\end{proof}
}

Only the matrices $K$ and $S$ depend explicitly on $\bm{\alpha}$ in the discrete Lagrangian. When
the saddle point conditions \eqref{eq_forward} and \eqref{eq_backward} are satisfied, we can
therefore calculate the gradient of ${\cal G}_h$ by differentiating \eqref{eq_disc-Lagrangian},
\[
\frac{\p {\cal G}_h}{\p \alpha_r} = \frac{\p {\cal L}_h}{\p \alpha_r},\quad r=0,1,\ldots,E-1.
\]
This relation leads to the following lemma.
\begin{lemma}\label{lem_adjoint-disc}
  Let ${\cal L}_h$ be the discrete Lagrangian defined by \eqref{eq_disc-Lagrangian}. Assume that
  $(\bm{u}_j^n, \bm{v}_j^n, \bm{U}_j^{n,i}, \bm{V}_j^{n,i})$ are calculated according to the
  St\"ormer-Verlet scheme for a given parameter vector $\bm{\alpha}$.  Furthermore, assume that
  $(\bm{\mu}_j^n,\bm{\nu}_j^n, \bm{X}_j^{n}, \bm{Y}_j^{n,i})$ satisfy the adjoint time-stepping
  scheme in Lemma \ref{lem_adjoint-bck}, subject to the terminal conditions
  \[
  \bm{\mu}_j^{M}  = -\frac{2}{E^2}\left( \text{Re}(S_{Vh})\bm{d}^u_j -
  \text{Im}(S_{Vh}) \bm{d}_j^v  \right), \quad
  \bm{\nu}_j^{M} = -\frac{2}{E^2}\left( \text{Re}(S_{Vh})\bm{d}^v_j +
  \text{Im}(S_{Vh}) \bm{d}_j^u  \right),
  \]
  and the forcing functions
  \begin{align*}
  \frac{\p {\cal J}_h}{\p \bm{U}_j^{n,1}} &= \frac{h}{T} W
  \bm{U}^{n,1}_j, \quad
  \frac{\p {\cal J}_h}{\p \bm{U}_j^{n,2}} = \frac{h}{T} W
  \bm{U}^{n,2}_j \\
  \frac{\p {\cal J}_h}{\p \bm{V}_j^{n,1}} &= \frac{h}{T} W
  \bm{V}^{n,1}_j,\quad
  \frac{\p {\cal J}_h}{\p \bm{V}_j^{n,2}} = 0.
  \end{align*}
  Then, the saddle-point conditions \eqref{eq_forward} and \eqref{eq_backward} are satisfied and the
  gradient of the objective function \eqref{eq:FullyDiscreteObjective} is given by
  \begin{align}
    \frac{\p {\cal G}_h}{\p \alpha_r} & = \frac{h}{2} \sum_{j=0}^{E-1} \sum_{n=0}^{M-1} \left\{  \left\langle S'_n\bm{U}_j^{n,1} 
    + S'_{n+1}\bm{U}_j^{n,2} - (K'_n + K'_{n+1}) \bm{V}_j^{n,1}, \bm{X}_j^{n} \right\rangle_2\, \right.\nonumber\\
    & \left. + \left\langle K'_{n+1/2} \bm{U}_j^{n,1} + S'_{n+1/2} \bm{V}_j^{n,1},
    \bm{Y}_j^{n,1}\right\rangle_2
    +  \left\langle K'_{n+1/2} \bm{U}_j^{n,2} + S'_{n+1/2} \bm{V}_j^{n,1},\bm{Y}_j^{n,2} \right\rangle_2  \right\}, \label{eq_L-grad}
  \end{align}
  where $S'_n = \p S/\p \alpha_r (t_n)$, $K'_{n+1/2}= \p K/\p \alpha_r (t_{n+1/2})$, etc.
\end{lemma}
\begin{proof}
See~\ref{appendix:adjointGrad}.
\end{proof}
As a result of Lemma~\ref{lem_adjoint-disc}, all components of the gradient can be calculated from
$(\bm{u}_j^n, \bm{v}_j^n, \bm{U}_j^{n,i}, \bm{V}_j^{n,1})$ and the adjoint variables
$(\bm{\mu}_j^n,\bm{\nu}_j^n, \bm{X}_j^{n},\bm{Y}_j^{n,i})$.  The first set of variables are obtained
from time-stepping the St\"ormer-Verlet scheme forward in time, while the second set of variables
follow from time-stepping the adjoint scheme backward in time.

We can avoid storing the time-history of $(\bm{u}^n_j, \bm{v}^n_j, \bm{U}_j^{n,i}, \bm{V}_j^{n,1})$
by using the time-reversibility of the St\"ormer-Verlet scheme. However, in order to do so, we must
first calculate the terminal conditions $(\bm{u}_j^{M},\bm{v}_j^{M})$ by evolving
\eqref{eq_initial-cond}-\eqref{eq_V2} forwards in time. The time-stepping can then be reversed and
the gradient of the objective function \eqref{eq_L-grad} can be accumulated by simultaneously
time-stepping the adjoint system \eqref{eq_mu}-\eqref{eq_Y2} backwards in time.

\section{Quadratic B-splines with carrier waves}\label{sec:B-Splines}

{\changed
Let $A(t)$ and $\phi(t)$ be real-valued amplitude and phase functions of time. By taking the control
functions in the rotating frame Hamiltonian \eqref{eq:hamiltonian-total} to be
\begin{align*}
  p(t) = A(t) \cos(\phi(t)),\quad  q(t) = A(t) \sin(\phi(t)),
\end{align*}
the relation \eqref{eq_rot-ansatz} results in the laboratory frame control function
\begin{equation}\label{eq_control_lab}
  f(t) =
  2 A(t)\cos(\omega_a t + \phi(t)).
\end{equation}
We expand the amplitude function in a set of basis functions $\{B_k\}_{k=1}^{D_1}$ and
start by considering the case of one carrier wave. We make the ansatz,
\[
A(t) = \sum_{k=1}^{D_1} B_k(t) \beta_k,
\]
where $\beta_k$ are real coefficients. By defining the phase as $\phi_k(t) = \Omega t + \theta_k$,
\begin{align*}
  p(t) &=  \sum_{k=1}^{D_1} B_k(t) \beta_k \cos( \Omega t + \theta_k) =
  \sum_{k=1}^{D_1} B_k(t) \left[ \alpha^{(1)}_k \cos(\Omega t) - \alpha^{(2)}_k \sin(\Omega t) \right],\\
  q(t) &=  \sum_{k=1}^{D_1} B_k(t) \beta_k \sin( \Omega t +\theta_k) =
  \sum_{k=1}^{D_1} B_k(t) \left[ \alpha^{(1)}_k \sin(\Omega t) + \alpha^{(2)}_k \cos(\Omega t) \right].
\end{align*}
where
\[
\alpha^{(1)}_k = \beta_k \cos(\theta_k),\quad   \alpha^{(2)}_k = \beta_k \sin(\theta_k).
\]
In the laboratory frame, the resulting control function becomes
\begin{align*}
  f(t) = 2 \sum_{k=1}^{D_1} B_k(t) \beta_k \cos((\omega_a +  \Omega) t + \theta_k).
%
\end{align*}

The case with one carrier wave is straightforward to generalize to multiple frequencies,
$\{\Omega_\ell\}_{\ell=1}^{N_f}$. This leads to a laboratory frame control function with a
spectrum that can be precisely specified to match the transition frequencies of the system,
%
\begin{align*}
f(t) = 2 \sum_{\ell=1}^{N_f} \sum_{k=1}^{D_1} B_k(t) \beta_{k,\ell} \cos((\omega_a +  \Omega_\ell) t
+ \theta_{k,\ell}).
\end{align*}
The total number of control parameters becomes $D = 2 N_f D_1$, which
equals the size of the parameter vector $\bm{\alpha}$. Here, $N_f$ is the number of frequencies and
$D_1\geq 1$ is the number of basis functions per frequency.
}

In this paper we use the quadratic B-spline basis (see Figure~\ref{fig_bspline}) to represent the
amplitude and phase of the control functions.
\begin{figure} 
    \centering
    \includegraphics[width=0.6\linewidth]{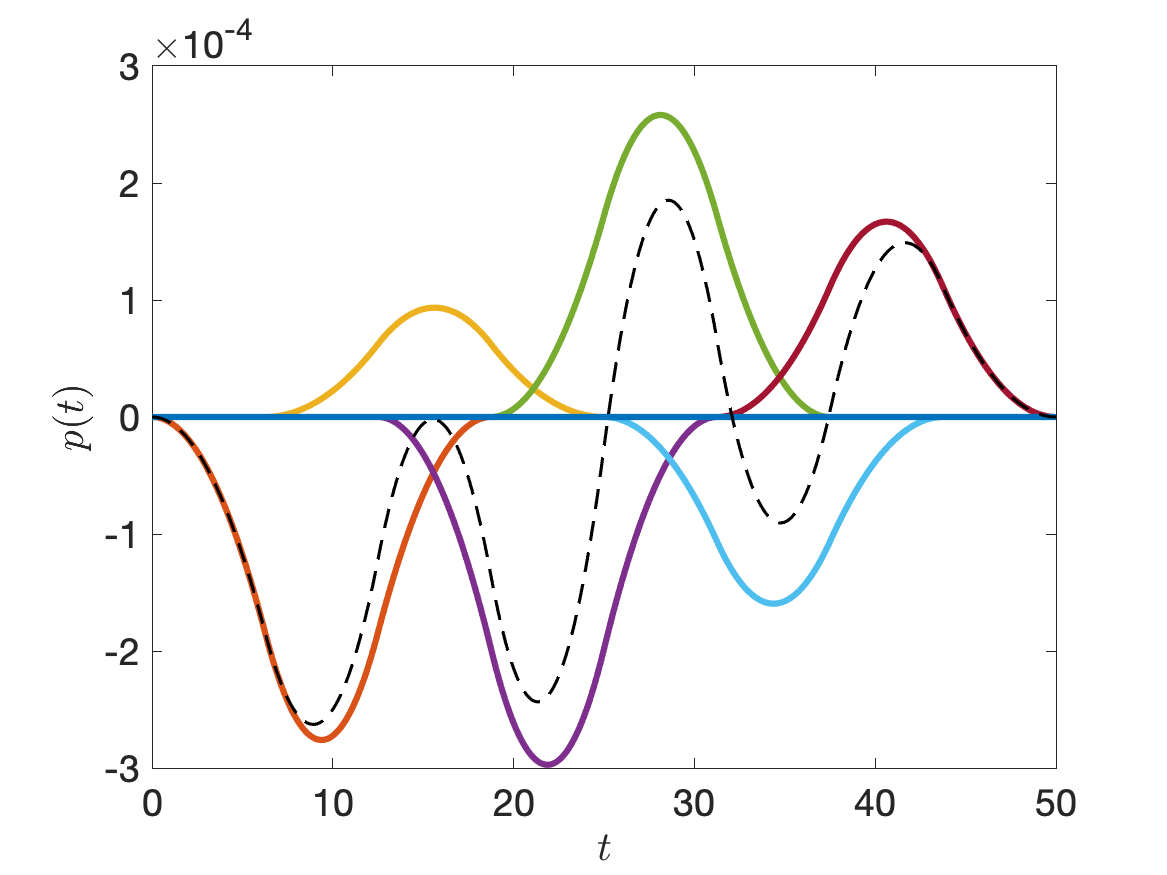}
    \caption{A B-spline control function $p(t,\bm{\alpha})$ without carrier wave ($\Omega_1=0$ and
      $N_f=1$). Here, the black dashed line is the control function and the solid colored lines are
      the individual B-spline basis functions, scaled by $\alpha^{1}_{m,1}$. In this case,
      $D_1=6$.}  \label{fig:spline_signal}
\end{figure}
Here, each basis function is a piecewise
quadratic polynomial in time. It is the lowest order B-spline function that has at least one
continuous derivative. We define the basis functions on a uniform grid in time,
\begin{equation}\label{eq_b-grid}
t_m = (m-1.5)\delta,\quad m=1,\ldots,D_1,\quad \delta = \frac{T}{D_1-2}.
\end{equation}
Each basis function $B_m(t)$ is centered around $t=t_m$ and is easily expressed in terms of the
scaled time parameter $\tau_m(t) = (t-t_m)/3\delta$,
\begin{align}\label{eq:splinebasis}
  B_m(t) = \widetilde{B}(\tau_m(t)),\quad
  \widetilde{B}(\tau) = \begin{cases}
    \frac{9}{8} + \frac{9}{2} \tau + \frac{9}{2} \tau^2, \quad & -\frac{1}{2} \leq \tau < -\frac{1}{6}, \\
    \frac{3}{4} - 9 \tau^2, \quad  & -\frac{1}{6} \leq \tau < \frac{1}{6}, \\
    \frac{9}{8} - \frac{9}{2} \tau + \frac{9}{2} \tau^2, \quad & \hphantom{-} \frac{1}{6} \leq \tau < \frac{1}{2}, \\
    0, \quad & \mbox{otherwise}.
  \end{cases}
\end{align}
Note that $B_m(t)$ is only non-zero in the interval $t \in [t_m -1.5\delta, t_m +1.5\delta]$. Thus,
for any fixed time $t$, a control function will only get contributions from at most three B-spline
basis functions. This property allows the control functions to be evaluated very efficiently.

\section{Numerical optimization}\label{sec_numopt}

Our numerical solution of the optimal control problem is based on the general purpose interior-point
optimization package IPOPT~\cite{Wachter2006}. This open-source library implements a primal-dual
barrier approach for solving large-scale nonlinear programming problems, i.e., it minimizes an
objective function subject to inequality (barrier) constraints on the parameter vector. Because the
Hessian of the objective function is costly to calculate, we use the L-BFGS
algorithm~\cite{Nocedal-Wright} in IPOPT, which only relies on the objective function and its
gradient to be evaluated. Inequality constraints that limit the amplitude of the parameter vector
$\bm{\alpha}$ are enforced internally by IPOPT.

The routines for evaluating the objective function and its gradient are implemented in the Julia
programming language~\cite{julia}, which provides a convenient interface to IPOPT. Given a parameter
vector $\bm{\alpha}$, the routine for evaluating the objective function solves the Schr\"odinger
equation with the St\"ormer-Verlet scheme and evaluates ${\cal G}_h(\bm{\alpha})$ by
accumulation. The routine for evaluating the gradient first applies the St\"ormer-Verlet scheme to
calculate terminal conditions for the state variables. It then proceeds by accumulating the gradient
$\nabla_\alpha{\cal G}_h$ by simultaneous reversed time-stepping of the discrete adjoint scheme and
the St\"ormer-Verlet scheme.
{\changed   These two fundamental routines, together with functions for
  setting up the Hamiltonians, estimating the time step, setting up constraints on the parameter
  vector, post-processing and plotting of the results have been implemented in the software package
  JuQBox, which was used to generate the numerical results below.

  The adjoint gradient implementation has been verified against a centered finite difference
  approximation of the discrete objective function by perturbing each component of the parameter
  vector. To further verify our implementation, we also calculated the discrete gradient by differentiating the
  St\"ormer-Verlet scheme with respect to each component of the parameter vector. This gradient
  agreed with the adjoint gradient to within $11\text{-}12$ digits. (Data not shown to conserve space.)
}

\subsection{A CNOT gate on a single qudit with guard levels}

To test our methods on a quantum optimal control problem, we consider realizing a CNOT gate on a
single qudit with four essential energy levels and two guard levels.  {\changed The qudit is modeled
  in the rotating frame of reference using the Hamiltonian \eqref{eq:hamiltonian-total} with
  fundamental frequency $\omega_a/2\pi= 4.10336$ GHz and self-Kerr coefficient $\xi_a/2\pi= 0.2198$
  GHz.}  We parameterize the two control functions using B-splines with carrier waves and choose the
frequencies to be $\Omega_1=0$, $\Omega_2 = -\xi_a$ and $\Omega_3 = -2 \xi_a$. In the rotating
frame, these frequencies correspond to transitions between the ground state and the first exited
state, the first and second excited states and the second and third excited states. We discourage
population of the fourth and fifth excited states using the weight matrix
$W=\text{diag}[0,0,0,0,0.1,1.0]$ in ${\cal J}_{2h}$, see \eqref{eq:objf-guard}. We use $D_1=10$
basis functions per frequency and control function, resulting in a total of $D=60$ parameters. The
amplitudes of the control functions are limited by the constraint
\begin{equation}\label{eq_amp-const}
  \|\bm{\alpha}\|_\infty := \max_{1\leq r\leq D}|\alpha_r|\leq \alpha_{max}.
\end{equation}
{\changed We set the gate duration to $T=100$ ns and} estimate the time
step using the technique in Section~\ref{sec_time-step-est}. To guarantee at least $C_P=40$ time
steps per period, we use {\changed $M=8,796$ time steps, corresponding to $h\approx 1.136\cdot
  10^{-2}$ ns.}

As initial guess for the elements of the parameter vector, we use a random number generator with a
uniform distribution in $[-0.01, 0.01]$.  {\changed In Figure~\ref{fig_ipopt-convergence} we present
  the convergence history with the two parameter thresholds $\alpha_{max}/2\pi = 4$ MHz and 3 MHz,
  respectively. We show the objective function ${\cal G}$, decomposed into ${\cal J}_{1h}$ and
  ${\cal J}_{2h}$, together with the norm of the dual infeasibility, $\|\nabla_\alpha{\cal G} -
  z\|_\infty$, that IPOPT uses to monitor convergence, see~\cite{Wachter2006} for details. For the
  case with $\alpha_{max}/2\pi = 3$ MHz, IPOPT converges well and needs 126 iteration to reduce the
  dual infeasibility to $10^{-5}$, which was used as convergence criteria. However, when the
  parameter constraint is relaxed to $\alpha_{max}/2\pi = 4$ MHz, the convergence of IPOPT stalls
  after about 100 iterations and is terminated after 200 iterations.}
\begin{figure}
  \centering
  \includegraphics[width=0.49\textwidth]{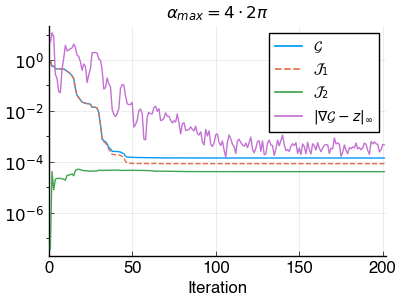}
  \includegraphics[width=0.49\textwidth]{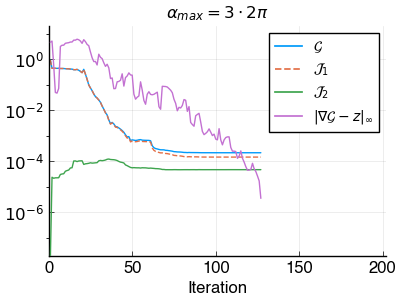}
  \caption{Convergence of the IPOPT iteration for the CNOT gate with the parameter constraint
    $\|\alphab\|_\infty \leq \alpha_{max}$. Here, $\alpha_{max}/2\pi = 4$ MHz (left) and $\alpha_{max}/2\pi = 3$ MHz (right).}
  \label{fig_ipopt-convergence}
\end{figure}

{\changed For the converged solution with parameter constraint $\alpha_{max}/2\pi = 3$ MHz, the two parts of the objective
  function are ${\cal J}_{1h} \approx 1.47\cdot 10^{-4}$ and ${\cal J}_{2h} \approx 4.72\cdot
  10^{-5}$, corresponding to a trace fidelity greater than $0.9998$.
  The population of the guard states remains small for all times and initial conditions. In
  particular, the ``forbidden'' state $|5\rangle$ has a population 
  that remains below $4.04\cdot 10^{-7}$, see Figure~\ref{fig:CNOT-forb}.}
\begin{figure}
  \centering
  \includegraphics[width=0.8\textwidth]{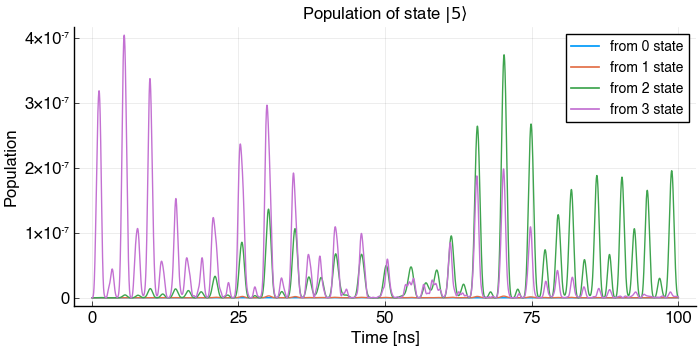}
  \caption{The population of the ``forbidden'' state $|5\rangle$ as function of time for the four
    initial conditions of the CNOT gate. Here, $\alpha_{max}/2\pi = 3$ MHz.}
  \label{fig:CNOT-forb}
\end{figure}
The optimized control functions are shown in Figure~\ref{fig:CNOT-ctrl} and the population of the
essential states, corresponding to the four initial conditions of the CNOT gate, are presented in
Figure~\ref{fig:CNOT-prob}.
\begin{figure}
  \centering
  \includegraphics[width=0.8\textwidth]{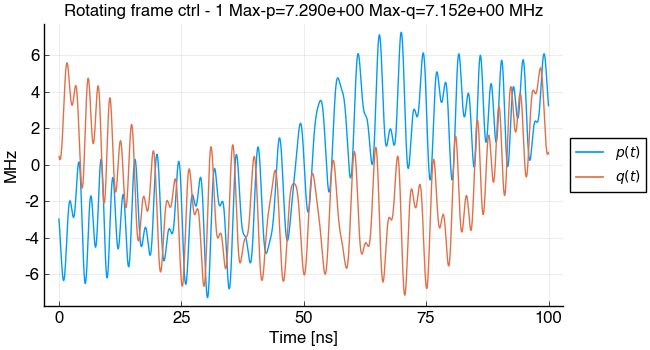}
  \caption{The rotating frame control functions $p(t)$ (blue) and $q(t)$ (orange) for realizing a CNOT gate with
    $D_1=10$ basis function per carrier wave and three carrier wave frequencies.  Here, $\alpha_{max}/2\pi = 3$ MHz.}
  \label{fig:CNOT-ctrl}
\end{figure}
\begin{figure}
  \centering
  \includegraphics[width=0.95\textwidth]{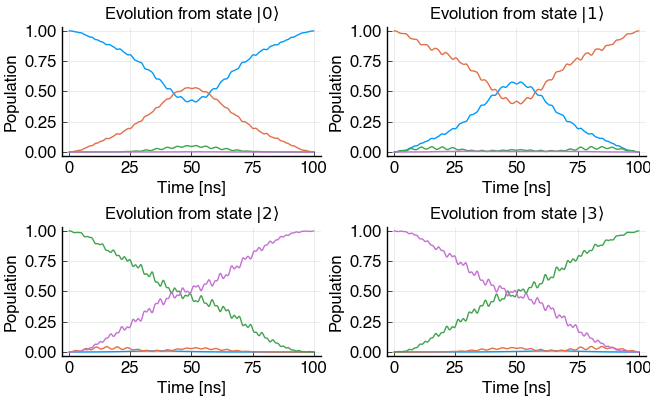}
  \caption{The population of the states $|0\rangle$ (blue), $|1\rangle$ (orange), $|2\rangle$
    (green) and $|3\rangle$ (purple), as function of time, for each initial condition of the
    CNOT gate. Here, $\alpha_{max}/2\pi = 3$ MHz.}
  \label{fig:CNOT-prob}
\end{figure}

{\changed Even though the dual infidelity does not reach the convergence criteria with the parameter
  threshold $\alpha_{max}/2\pi = 4$ MHz, the resulting control functions give a very small objective
  function. Here,
  ${\cal J}_{1h} \approx 8.56\cdot 10^{-5}$ and ${\cal J}_{2h} \approx 4.15\cdot 10^{-5}$,
  corresponding to a trace fidelity greater than $0.9999$.  The population of the ``forbidden'' state
  $|5\rangle$ has a population that remains below $3.39\cdot 10^{-7}$.}


\subsection{The Hessian of the objective function}
{\changed The numerical results shown in Figure~\ref{fig_ipopt-convergence} illustrate that the
  convergence properties of the optimization algorithm depend on the parameter constraints. To gain
  clarity into the local landscape of the optima we study the Hessian of the objective function.}
Let the optima correspond to the parameter vector $\bm{\alpha}^*$. Based on the adjoint scheme for
calculating the gradient, we can approximate the elements of the Hessian matrix using a centered
finite difference approximation,
\begin{equation}
\label{eq_hess-fd}
\frac{\p^2 {\cal G}_{h}(\bm{\alpha}^*)}{\p \alpha_j \p \alpha_k} \approx
\frac{1}{2\varepsilon}\left(
\frac{\p {\cal G}_{h}}{\p \alpha_j}(\bm{\alpha}^*+\varepsilon \bm{e}_k) -
\frac{\p {\cal G}_{h}}{\p \alpha_j}(\bm{\alpha}^*-\varepsilon \bm{e}_k)
\right) := L_{j,k},
\end{equation}
for $j,k=1,2,\ldots,D$. To perform this calculation, the gradient must be evaluated for the $2D$
parameter vectors $(\bm{\alpha}^*\pm\varepsilon \bm{e}_k)$. Because the objective function and the
parameter vector are real-valued, the gradient and the Hessian are also real-valued. Due to the
finite difference approximation, the matrix $L$ is only approximately equal to the Hessian. The
accuracy in $L$ is estimated in Table~\ref{tab_fd-hess} by studying the norm of its asymmetric
part, which is zero for the Hessian.
\begin{table}
\centering
\begin{tabular}{c|c|c} 
      $\varepsilon$ & $\| 0.5(L+L^T) \|_F$ & $\| 0.5(L - L^T) \|_F$ \\
      \hline
      $10^{-4}$ & $4.95\cdot 10^{3}$ & $1.99\cdot 10^{-4}$ \\
      $10^{-5}$ & $4.95\cdot 10^{3}$ & $2.01\cdot 10^{-6}$ \\
      $10^{-6}$ & $4.95\cdot 10^{3}$ & $1.46\cdot 10^{-6}$ \\
      $10^{-7}$ & $4.95\cdot 10^{3}$ & $1.47\cdot 10^{-5}$ \\
\end{tabular}
\caption{The Frobenius norm of the symmetric and asymmetric parts of the approximate Hessian, $L$,
  for the case $\alpha_{max}/2\pi = 3.0$ MHz.}\label{tab_fd-hess}
\end{table}
Based on this experiment we infer that $\varepsilon=10^{-6}$ is appropriate to use for approximating
the Hessian in \eqref{eq_hess-fd}. To eliminate spurious effects from the asymmetry in the $L$
matrix, we study the spectrum of its symmetric part, $L_{s} = 0.5(L+L^T)$. Because it is real and
symmetric, it has a complete set of eigenvectors and all eigenvalues are real.

The eigenvalues of the Hessian are shown in Figure~\ref{fig:CNOT-spectrum} for both values of the
parameter threshold, $\alpha_{max}$.
\begin{figure}
  \centering
  \includegraphics[width=0.49\textwidth]{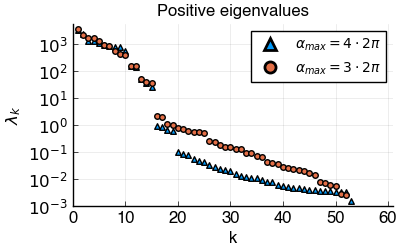}
  \includegraphics[width=0.49\textwidth]{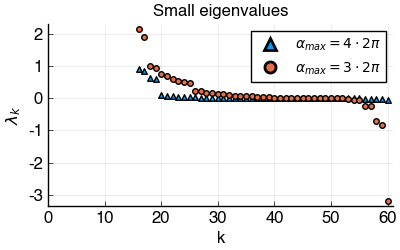}
  \caption{The eigenvalues of the symmetric part of the approximate Hessian, $0.5(L+L^T)$, evaluated
    at the optima for the parameter thresholds {\changed $\alpha_{max}/2\pi = 4$ MHz (blue
      triangles) and $\alpha_{max}/2\pi = 3$ MHz (orange circles)}. The positive eigenvalues are
    shown on a log-scale on the left and the small eigenvalues are shown on a linear scale on the
    right.}
  \label{fig:CNOT-spectrum}
\end{figure}
{\changed Two properties of the spectra are noteworthy. First, a few eigenvalues are negative. This
  may be an artifact related to the elements of the parameter vector that are close to their
  bounds. As a result the landscape of the objective function may not be accurately represented by
  the corresponding components of the Hessian.  The second interesting property is
  that the 15 largest eigenvalues are significantly larger than the rest. This indicates that the
  control functions are essentially described by the 15 eigenvectors associated with those
  eigenvalues. As a result, the objective function varies much faster in those directions than in the
  directions of the remaining 45 eigenvectors and this may hamper the convergence of the optimization
  algorithm in that subspace. However, most of those 45 eigenvalues become larger when the parameter
  threshold is reduced from $\alpha_{max}/2\pi = 4$ MHz to $\alpha_{max}/2\pi = 3$ MHz. This
  indicates that the constraints on the parameter vector have a regularizing effect on the
  optimization problem and may explain why the latter case converges better (see
  Figure~\ref{fig_ipopt-convergence}).}


{\changed
\section{Comparing JuQBox with QuTiP/pulse\_optim and Grape-TF}\label{sec_compare}

The QuTiP/pulse\_optim package is part of the QuTiP~\cite{qutip} framework and implements the GRAPE
algorithm in the Python language. The Grape-TF code (TF is short for
Tensorflow~\cite{abadi2016tensorflow}) is also implemented in Python and provides an enhanced
implementation of the GRAPE algorithm, as described by Leung et al.~\cite{Leung-2017}. It is
callable from QuTiP and shares a similar problem setup with the pulse\_optim function.

To compare the JuQBox code with pulse\_optim and Grape-TF, we consider a set of SWAP gates. These
gates transform the ground state $|0\rangle$ to excited state $|d\rangle$, and vice versa. The
transformation can be described by the unitary matrix
\begin{align}\label{eq_essential-target}
  V_g =
  \begin{bmatrix}
    0 & 0 & \cdots & 0 & 1 \\
    0 & 1 & 0 & \cdots & 0 \\
    \vdots & \ddots & \ddots & \ddots & \vdots \\
     0 & \cdots & 0 & 1 & 0 \\
    1 & 0 & \cdots & 0 & 0
  \end{bmatrix} \in \mathbb{C}^{(d+1)\times (d+1)},
\end{align}
which involves $E=d+1$ essential states. To evaluate how much leakage occurs to higher energy
levels, we add one guard (forbidden) level ($G=1$) and evolve a total of $N=d+2$ states in
Schr\"odinger's equation. As before, the guard level is left unspecified in the target gate
transformation. We consider implementing the SWAP gates on a multi-level qudit that can be described
by the fundamental frequency $\omega_a/2\pi = 4.8$ GHz and the self-Kerr coefficient $\xi_a/2\pi =
0.22$ GHz. We apply the rotating wave approximation, where the angular frequency of the rotation is
$\omega_a$, resulting in the Hamiltonian model \eqref{eq:hamiltonian-total}. As a realistic model
for current superconducting quantum devices, we impose the control amplitude restrictions
\begin{align} \label{eq_amp-bounds}
  \mbox{max}_t|p(t,\alphab)| \leq c_\infty,\quad
  \mbox{max}_t|q(t,\alphab)| \leq c_\infty,\quad
  \frac{c_\infty}{2\pi} = 9\,\mbox{MHz},
\end{align}
in the rotating frame of reference.

\subsection{Setup of simulation codes}
QuTiP/pulse\_optim can minimize the target gate fidelity, ${\cal G}_1$, but does
not suppress occupation of higher energy states. Thus, it does {\em not} minimize terms of the type
${\cal G}_2$.  As a proxy for ${\cal G}_2$, we append one additional energy level to the simulation and
measure its occupation as an estimate of leakage to higher energy states.  In pulse\_optim, the
control functions are discretized on the same grid in time as Schr\"odinger's equation and no
smoothness conditions are imposed. In our tests, we use a random initial guess for the parameter
vector.

Grape-TF discretizes the control functions on the same grid in time as Schr\"odinger's equation. It
minimizes an objective function that consists of a number of user-configurable parts. In our test,
we minimize the gate infidelity (${\cal G}_1$) and the occupation of one guard (forbidden) energy
level (similar to ${\cal G}_2$). To smooth the control functions in time, the objective function
also contains additional terms to minimize their first and second time derivatives. The various
parts of the objective function are weighted together by user-specified coefficients. The gradient
of the objective function is calculated using the automatic differentiation (AD) technique, as
implemented in the Tensorflow package. In our tests, we use a random initial guess for the control
vector.

In JuQBox, we trigger the first $d$ transition frequencies in the Hamiltonian by using $d$ carrier
waves in the control functions, with angular frequencies
\begin{align*}
    \Omega_k = (k-1) (-\xi_a),\quad k=1,2,\ldots,N_f,\quad N_f = d.
\end{align*}
Similar to pulse\_optim and Grape-TF, a pseudo-random number generator is used to construct the
initial guess for the parameter vector.

The pulse\_optim and JuQBox simulations were run on a Macbook Pro with a 2.6 GHz Intel iCore-7
processor. To utilize the GPU acceleration in Tensorflow, the Grape-TF simulations were run on one
node of the Pascal machine at Livermore Computing, where each node has an Intel XEON E5-2695 v4
processor with two NVIDIA P-100 GPUs.

\subsection{Numerical results}

A SWAP gate where the control functions meet the control amplitude bounds \eqref{eq_amp-bounds} can
only be realized if the gate duration is sufficiently long. Furthermore, the minimum gate duration
increases with $d$. For each value of $d$, we used numerical experiments to determine a duration
$T_d$ such that at least two of the three simulation codes could find a solution with a small gate
infidelity. For JuQBox, we used the technique in Section~\ref{sec_time-step-est} with $C_P=80$ to
obtain the number of time steps. The number of control parameters follow from $D=2N_f D_1$, where
$N_f=d$ equals the number of carrier wave frequencies and $D_1$ is the number of B-splines per
control functions. Here, $D_1=10$ for $d=3,4,5$ and $D_1=20$ for $d=6$. For pulse\_optim and
Grape-TF, we calculate the number of time steps based on the shortest transition period,
corresponding to the highest transition frequency in the system. We then use 40 time steps per
shortest transition period to resolve the control functions. For both GRAPE methods there are 2
control parameters per time step. The main simulation parameters are given in
Table~\ref{tab_params}.
\begin{table}[tph]
    \begin{center}
\begin{tabular}{r|r||r|r||r|r}
\multicolumn{2}{c||}{} & \multicolumn{2}{c||}{\# time steps} & \multicolumn{2}{c}{\# parameters} \\ \hline
$d$ & $T_d$ [ns] & JuQBox & GRAPE & JuQBox & GRAPE\\ \hline
3   & 140        & 14,787  & 4,480 & 60 & 8,960 \\ \hline 
4   & 215        & 37,843  & 7,568 & 80  & 15,136 \\ \hline 
5   & 265        & 69,962  & 11,661 & 100 & 23,322 \\ \hline 
6   & 425        & 157,082 & 22,441 & 240 & 44,882 \\ \hline 
\end{tabular}
\caption{Gate duration, number of time steps ($M$) and total number of control parameters ($D$) in
  the $|0\rangle \leftrightarrow |d\rangle$ SWAP gate simulations. The number of time steps and
  control parameters are the same for pulse\_optim and Grape-TF.}\label{tab_params}
\end{center}
\end{table}

Optimization results for the pulse\_optim, Grape-TF and JuQBox codes are presented in
Tables~\ref{tab_qutip}, \ref{tab_grape} and~\ref{tab_juqbox}.
The pulse\_optim code generates piecewise constant control functions that are very noisy and may
therefore be hard to realize experimentally. To obtain a realistic estimate of the resulting
dynamics, we interpolate the optimized control functions on a grid with 20 times smaller time step
and use the \verb|mesolve()| function in QuTiP to calculate the evolution of the system from each
initial state. We then evaluate the gate infidelity using the evolved states at the final time,
denoted by ${\cal G}_1^*$ in Table~\ref{tab_qutip}. Since the control functions from Grape-TF and
JuQBox are significantly smoother, we report the target gate fidelities as calculated by those
codes.
\begin{table}[tph]
    \begin{center}
\begin{tabular}{r|r|r|r|r|r|r}
  $d$ & ${\cal G}_1^*$ & $|\psi^{(d+1)}|^2_\infty$ & $|p|_\infty$ [MHz]& $|q|_\infty$ [MHz] & \# iter & CPU [s]
  \\ \hline
  3 & 4.35e-6 & 9.41e-3 & 9.00 & 9.00 & 38 & 30 \\ \hline
  4 & 3.91e-5 & 1.20e-2   & 9.00 & 9.00 & 93 & 108 \\ \hline
  5 & 1.57e-4 & 8.77e-3   & 9.00 & 9.00 & 215 & 385 \\ \hline
  6 & {\bf 1.76e-3} & {\bf 4.48e-2} & 9.00 & 9.00 & 246 & 894 \\ \hline
\end{tabular}
\caption{QuTiP/pulse\_optim results for $|0\rangle\leftrightarrow |d\rangle$ SWAP gates. Note the
  larger infidelity and guard state population for $d=6$. }\label{tab_qutip}
    \end{center}
\end{table}
\begin{table}[thp]
    \begin{center}
\begin{tabular}{r|r|r|r|r|r|r}
  $d$ & ${\cal G}_1$ & $|\psi^{(d+1)}|^2_\infty$ &  $|p|_\infty$ [MHz] & $|q|_\infty$ [MHz] & \# iter & CPU [s]
  \\ \hline
  3 & 8.76e-6 & 4.03e-3 & 6.98  & 8.83 & 78 & 2,062 \\ \hline
  4 & 1.52e-5 & 3.39e-3 & 6.87 & 6.54 & 128 & 10,601  \\ \hline
%
  5 & 2.80e-5 & 1.78e-3 & 7.21 & 7.62 & 161 & 28,366 \\ \hline
  6 & {\bf 4.89e-1} & 2.33e-5 & 0.73 & 0.74 & 93  & 81,765 \\ \hline
\end{tabular}
\caption{Grape-TF results for $|0\rangle\leftrightarrow |d\rangle$ SWAP gates. Note the very large
  infidelity for $d=6$. These simulations used two NVIDIA P-100 GPUs to accelerate
  Tensorflow.}\label{tab_grape}
    \end{center}
\end{table}
\begin{table}[thp]
\begin{center}
\begin{tabular}{r|r|r|r|r|r|r}
  $d$ & ${\cal G}_1$ & $|\psi^{(d+1)}|^2_\infty$ & $|p|_\infty$ [MHz]& $|q|_\infty$ [MHz]& \# iter & CPU
  \\ \hline
  3 & 2.71e-5 & 1.92e-3 & 7.59 & 8.99 & 177 & 55 \\ \hline
  4 & 4.91e-5 & 1.23e-3 & 7.78 & 5.33 & 166 & 151 \\ \hline
  5 & 4.95e-5 & 1.25e-3 & 7.42 & 7.24 & 173 & 291 \\ \hline
  6 & 7.41e-6 & 4.41e-3 & 4.55 & 5.39 & 229 & 1255 \\ \hline
\end{tabular}
\caption{JuQBox results for $|0\rangle\leftrightarrow |d\rangle$ SWAP gates.}\label{tab_juqbox}
\end{center}
\end{table}

For the $|0\rangle \leftrightarrow |3\rangle$, $|0\rangle \leftrightarrow |4\rangle$ and $|0\rangle
\leftrightarrow |5\rangle$ SWAP gates, all three codes produce control functions with very small
gate infidelities. We note that the population of the guard level, $|\psi^{(d+1)}|^2$, is about an
order of magnitude larger with pulse\_optim than with JuQBox; the guard level population from
Grape-TF are somewhere in between. The most significant difference between the results occur for the
$d=6$ SWAP gate. Here, the Grape-TF code fails to produce a small gate infidelity after running for
almost 23 hours and the pulse\_optim code results in a gate fidelity that is about 2 orders of
magnitude larger than JuQBox.

While pulse\_optim and JuQBox require comparable amounts of CPU time to converge, the Grape-TF code
is between 50-100 times slower, despite the GPU acceleration.

We proceed by analyzing the optimized control functions and take the $|0\rangle \leftrightarrow
|5\rangle$ SWAP gate as a representative example. In this case, the relevant transition frequencies
in the laboratory frame of reference are
\begin{align}
f_k = \frac{1}{2\pi} \left( \omega_a - k \xi_a\right),\quad k=0,1,2,3,4.
\end{align}
To compare the optimized control functions, we evaluate the corresponding laboratory frame control
function using \eqref{eq_rot-ansatz} and study its Fourier spectrum. Results from the pulse\_optim,
Grape-TF and JuQBox simulations are presented in Figure~\ref{fig_ctrl-fft}. We first note that
pulse\_optim produces a significantly noisier control function compared to the other two codes. The
control function from Grape-TF is significantly smoother, even though its spectrum includes some
noticeable peaks at frequencies that do not correspond to transition frequencies in the system. The
JuQBox simulation results in a laboratory frame control function where each peak in the spectrum
corresponds to a transition frequency in the Hamiltonian.
\begin{figure}[htp]

  \begin{subfigure}{1.0\textwidth}
    \begin{center}
    \includegraphics[width=0.8\linewidth]{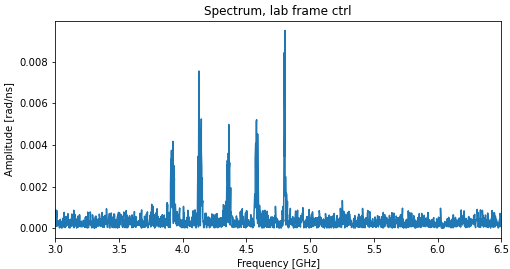}
    \caption{QuTiP/pulse\_optim.}
    \end{center}
\end{subfigure}
 
  \begin{subfigure}{1.0\textwidth}
    \begin{center}
  \vspace{3mm}
  \includegraphics[width=0.8\linewidth]{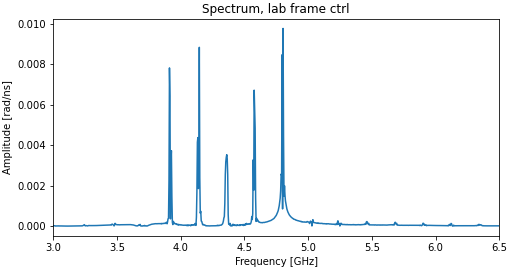}
  \caption{Grape-TF.}
     \end{center}
 \end{subfigure}
  
  \begin{subfigure}{1.0\textwidth}
   \begin{center}
   \vspace{3mm}  \includegraphics[width=0.8\linewidth]{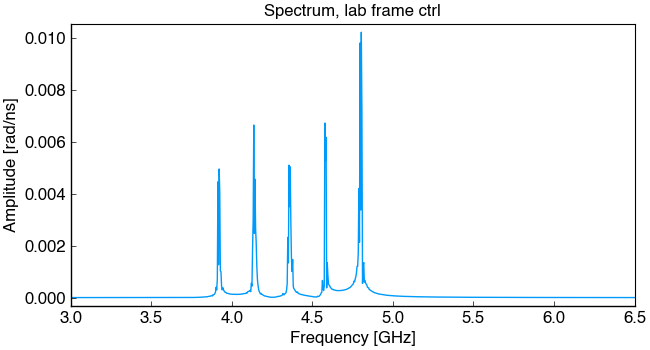}
    \caption{JuQBox.}
    \end{center}
  \end{subfigure}
\caption{Magnitude of the Fourier spectrum of the laboratory frame control function for the $|0\rangle \leftrightarrow |5\rangle$ SWAP gate.}\label{fig_ctrl-fft}    

\end{figure}
} 

\section{Conclusions}\label{sec_conc}

In this paper we have developed numerical methods for optimizing control functions for realizing
logical gates in a closed quantum system. {\changed The quantum state is governed by Schr\"odinger's equation,
which is a time-dependent Hamiltonian system. To ensure long-time numerical accuracy we discretize
it using the symplectic St\"ormer-Verlet method, which can be written as a partitioned Runge-Kutta
scheme.  Our main theoretical contribution is the derivation of a compatible time-discretization of
the adjoint state equation, such that the gradient of the discrete objective function can be
calculated exactly. This scheme generalizes Ober-Bl\"obaum's~\cite{ober2008discrete} methods to
the case of a time-dependent Hamiltonian system.}

We have also introduced a parameterization of the control functions based on B-splines with built-in
carrier waves. {\changed The carrier waves are used to specify the frequency spectra of the control functions,
while the B-spline functions specify their envelope and phase. This approach allows
the number of control parameters to be independent of, and significantly smaller than, the number of
time steps for integrating Schr\"odinger's equation.} Our numerical solution of the
optimal control problem is based on the general purpose interior-point optimization package
IPOPT~\cite{Wachter2006}, which implements a primal-dual barrier approach for minimizing the
objective function subject to amplitude constraints on the parameter vector. We optimized the
control functions for a CNOT gate with two guard states, resulting in a gate trace fidelity greater
than 0.9999. Having a moderate number of control parameters enabled us to study the spectrum of the
Hessian of the objective function at an optima. We found that imposing tighter bounds on the
parameter vector results in a Hessian with larger eigenvalues and thus improves the convergence of the
optimization algorithm.

{\changed The performance of the proposed algorithm, implemented in a code called JuQBox, was
  compared with two implementations of the GRAPE algorithm: QuTiP/pulse\_optim~\cite{qutip} and
  Grape-Tensorflow~\cite{Leung-2017}. JuQBox was found to produce significantly smoother control
  functions than QuTiP/pulse\_optim, while using about the same computational
  resources. JuQBox was also found to run about 50-100 times faster than Grape-Tensorflow.

  In future work, it would be interesting to study if the convergence properties of the optimization
  algorithm can be improved by modifying the objective function. We also intend to generalize
  our approach to solve optimal control problem for open quantum systems.}

\section*{Acknowledgment} 
We would like to thank Prof.~Daniel Appel\"o for bringing the St\"ormer-Verlet method to our attention.

This work was supported in part by LLNL laboratory directed research and development project
20-ERD-028 and in part by DOE office of advanced scientific computing research (OASCR) under the
Advanced Research in Quantum Computing (ARQC) program, award 2019-LLNL-SCW-1683.

This work performed under the auspices of the U.S. Department of Energy by Lawrence Livermore
National Laboratory under Contract DE-AC52-07NA27344. This is contribution LLNL-JRNL-800457.

\appendix

\section{The Hamiltonian in a rotating frame of reference}\label{app_RotatingFrame}

In the laboratory frame of reference, the Hamiltonian matrix for a single superconducting qudit can be modeled by
\begin{equation}\label{eq_quantum-osc-A}
H(t) = \omega_a a^\dag a - \frac{\xi_a}{2} a^\dag a^\dag a a + f(t)(a + a^\dagger).
\end{equation}
Here, $\omega_a>0$ and $\xi_a$ are given real constants and $f(t,\bm{\alpha})$ is a real-valued
function of time that depend on the parameter vector $\bm{\alpha}$.  Furthermore, $a$ is the
lowering matrix,
\[
a = \begin{bmatrix}
0 & 1 & & &  \\
& 0 & \sqrt{2} & &  \\
 & & \ddots & \ddots &\\
 & & & 0 & \sqrt{N-1}\\
 & & & & 0
\end{bmatrix},
\]
and the raising matrix $a^\dag$ is its adjoint (conjugate transpose).

To derive the rotating frame transformation, we consider the time-dependent change of variable
\[
\bm{\psi}(t) = R^{\dag}(t)\widetilde{\bm{\psi}}(t),\quad R^\dag R = I.
\]
We have
\[
\dot{\bm{\psi}} = \dot{R}^\dag \widetilde{\bm{\psi}} + R^\dag \dot{\widetilde{\bm{\psi}}},\quad
H\bm{\psi} = H R^\dag \widetilde{\bm{\psi}}.
\]
After some algebra, the Schr\"odinger equation \eqref{eq:schrodinger_example} and the identity $R \dot{R}^\dag = -
\dot{R} R^\dag$ gives:
\begin{equation}\label{eq_timedep_trans}
\dot{\widetilde{\bm{\psi}}} = -i \tilde{H}(t) \widetilde{\bm{\psi}},\quad \tilde{H}(t) = R(t)H(t)R(t)^\dag + i \dot{R}(t) R(t)^\dag.
\end{equation}
The rotating frame of reference is introduced by taking the unitary transformation to be
\begin{equation}\label{eq_rot-frame-trans}
R(t) = \exp(i \omega_a t\, a^\dag a ), \quad a^\dag a =
\begin{bmatrix}
  0 &&&& \\
  & 1 &&& \\
  && 2 &&  \\
  &&& \ddots & \\
  &&&& N-1
\end{bmatrix},\quad \dot{R}R^\dag = i\omega_a a^\dag a.
\end{equation}
From \eqref{eq_timedep_trans} and \eqref{eq_rot-frame-trans}, the first term in the Hamiltonian \eqref{eq_quantum-osc-A} is canceled by the term $i \dot{R}(t) R(t)^\dag$. Furthermore, $a^\dag a^\dag a a = (a^\dag a)^2 - a^\dag a$ and both $a^\dag a$ and $(a^\dag a)^2$ commute with $R(t)$. After noting that $R a^\dag R^\dag = e^{i\omega_a t} a^\dag$, the transformed Hamiltonian can be written
\begin{equation}\label{eq_trans-hamiltonian}
  \widetilde{H}(t) = -\frac{\xi_a}{2} \left((a^\dag a)^2 - a^\dag a \right) +
  f(t) \left( e^{-i\omega_a t} a + e^{i\omega_a t} a^\dag \right).
\end{equation}

To slow down the time scales in the control function, we want to absorb the highly oscillatory
factors $\exp(\pm i\omega_a t)$ into $f(t)$. Because the control function $f(t)$ is real-valued, this can only be done in an approximate fashion. We make the ansatz,
\begin{multline}
  f(t) = 2p(t) \cos(\omega_a t) - 2q(t) \sin(\omega_a t) = \\
  \left(p + i q\right)\exp(i\omega_a t) + \left( p - i q \right)
  \exp(-i\omega_a t),
\end{multline}
where $p(t)$ and $q(t)$ are real-valued functions. After some algebra, the transformed
Hamiltonian \eqref{eq_trans-hamiltonian} becomes
\begin{multline*}
  \widetilde{H}(t) 
  %
  = -\frac{\xi_a}{2} \left((a^\dag a)^2 - a^\dag a \right) + p\left( a + a^\dag \right) + i q \left( a - a^\dag \right)
  \\
  +  \left(p - iq\right) \exp(-2i\omega_a t) a  + \left(p + iq\right)\exp(2i\omega_a t) a^\dag.
\end{multline*}
The rotating frame approximation follows by ignoring the terms that oscillate with twice the
frequency, $\exp(\pm 2i\omega_a t)$, resulting in the transformed Schr\"odinger system,
\begin{align}\label{eq:rotschroedinger}
  \dot{\widetilde{\bm{\psi}}}_j &= - i\left(H_d + \widetilde{H}_c(t)\right)
  \widetilde{\bm{\psi}}_j,\quad
  \widetilde{\bm{\psi}}_j(0) = \bm{e}_j,\\
  H_d &= -\frac{\xi_a}{2} \left(a^\dag  a^\dag a a\right), \quad
  \widetilde{H}_c(t) =  p(t)\left( a + a^\dag \right) + i q(t) \left( a - a^\dag \right).
\end{align}
Here, $H_d$ is called the drift Hamiltonian. When $\xi_a \ll \omega_a$, the state vector varies on a
significantly slower time scale in the rotating frame than in the laboratory frame.

In the remainder of the paper, the Schr\"odinger equation is always solved under the rotating frame
approximation and we drop the tildes on the state vector and the Hamiltonian matrices.

\section{Derivation of the discrete adjoint scheme}\label{appendix:adjoint} 

We seek to determine a scheme for evolving the Lagrange multiplier (adjoint) variables to satisfy
the first order optimality conditions~\eqref{eq_backward}. In the following, let
$\delta_{r,s}$ denote the usual Kronecker delta function.

The terms $T^3_j$ to $T^6_j$ in \eqref{eq_disc-Lagrangian} enforce the relations between the stage
variables \eqref{eq_U1}-\eqref{eq_V2} according to
\begin{align}
T^3_j &= \sum_{n=0}^{M-1} \left\langle \bm{U}_j^{n,1} - \bm{u}_j^n, \bm{M}_j^{n,1} \right\rangle_2,\\
T^4_j &= \sum_{n=0}^{M-1} \left\langle \bm{U}_j^{n,2} - \bm{u}_j^n -
\frac{h}{2}\left(S_n\bm{U}_j^{n,1} + S_{n+1}\bm{U}_j^{n,2}  - K_n\bm{V}_j^{n,1} - K_{n+1}\bm{V}_j^{n,2} \right), \bm{M}_j^{n,2} \right\rangle_2,\\
T^5_j &= \sum_{n=0}^{M-1} \left\langle \bm{V}_j^{n,1} - \bm{v}_j^n - \frac{h}{2}\left(  K_{n+1/2}
\bm{U}_j^{n,1} + S_{n+1/2} \bm{V}_j^{n,1} \right), \bm{N}_j^{n,1} \right\rangle_2,\\
T^6_j &= \sum_{n=0}^{M-1} \left\langle \bm{V}_j^{n,2} - \bm{v}_j^n - \frac{h}{2}\left(  K_{n+1/2}
\bm{U}_j^{n,1} + S_{n+1/2} \bm{V}_j^{n,1} \right), \bm{N}_j^{n,2} \right\rangle_2.
\end{align}
Taking the derivative of \eqref{eq_disc-Lagrangian} with respect to $\bm{u}_j^r$
\begin{align*}
  0 = \pdiff{\Lag_h}{\bm{u}_j^r} & = \pdiff{{\cal J}_h}{\bm{u}_j^r}-\left[ (\bm{\mu}_j^n
    -\bm{\mu}_j^{n+1})\delta_{r,n} + \bm{\mu}_j^M\delta_{r,M} - (\bm{M}_j^{n,1} + \bm{M}_j^{n,2})
    \delta_{r,n}\right],   
\end{align*}
which gives the conditions
\begin{align*}
  \bm{\mu}_j^M = \pdiff{{\cal J}_h}{\bm{u}_j^M}, \quad \bm{\mu}_j^n - \bm{\mu}_j^{n+1} =
  \bm{M}_j^{n,1} + \bm{M}_j^{n,2}, \quad n=0,1,\ldots, M-1.
\end{align*}
Similarly, differentiating \eqref{eq_disc-Lagrangian} with respect to $\bm{v}_j^r$ gives
\begin{align*}
  0 = \pdiff{\Lag_h}{\bm{v}_j^r} & = \pdiff{{\cal J}_h}{\bm{v}_j^r}
  - \left[(\bm{\nu}_j^n -\bm{\nu}_j^{n+1})\delta_{r,n} + \bm{\nu}_j^M\delta_{r,M} - (\bm{N}_j^{n,1} + \bm{N}_j^{n,2}) \delta_{r,n}\right],
\end{align*}
which leads to the conditions
\begin{align*}
  \bm{\nu}_j^n - \bm{\nu}_j^{n+1}  = \bm{N}_j^{n,1} + \bm{N}_j^{n,2}, \quad \bm{\nu}_j^M  = \pdiff{{\cal J}_h}{\bm{v}_j^M}.
\end{align*}
Next we take the derivative of \eqref{eq_disc-Lagrangian} with respect to $\bm{U}_j^{n,1}$,
\begin{align*}
  \pdiff{\Lag_h}{\bm{U}_j^{n,1}} & = \pdiff{{\cal J}_h}{\bm{U}_j^{n,1}}-\sum_{i=1}^6 \pdiff{T_j^i}{\bm{U}_j^{n,1}} = 0, \\
  \pdiff{T_j^1}{\bm{U}_j^{n,1}} & = -\frac{h}{2}S_n^T\bm{\mu}_j^{n+1},\\
  \pdiff{T_j^2}{\bm{U}_j^{n,1}} & = -\frac{h}{2}K_{n+1/2}^T\bm{\nu}_j^{n+1}, \\
  \pdiff{T_j^3}{\bm{U}_j^{n,1}} & = \bm{M}_j^{n,1}, \\
  \pdiff{T_j^4}{\bm{U}_j^{n,1}} & = -\frac{h}{2}S_n^T\bm{M}_j^{n,2}, \\
  \pdiff{T_j^5}{\bm{U}_j^{n,1}} & = -\frac{h}{2}K_{n+1/2}^T \bm{N}_j^{n,1}, \\
  \pdiff{T_j^6}{\bm{U}_j^{n,1}} & = -\frac{h}{2}K_{n+1/2}^T \bm{N}_j^{n,2}, 
\end{align*}
which, using the fact that $S_n^T = -S_n$ and $K_n^T = K_n$, we may write as 
\begin{align*}
  \bm{M}_j^{n,1} + \frac{h}{2}S_n\left(\bm{\mu}_j^{n+1} + \bm{M}_j^{n,2}\right) - \frac{h}{2}
  K_{n+1/2}\left(\bm{\nu}_j^{n+1} + \bm{N}_j^{n,1} + \bm{N}_j^{n,2}\right) =
  \pdiff{{\cal J}_h}{\bm{U}_j^{n,1}}. 
\end{align*}
Repeating this procedure for the derivative with respect to $\bm{U}_j^{n,2}$ gives 
\begin{align*}
  \pdiff{\Lag_h}{\bm{U}_j^{n,2}} & = \pdiff{{\cal J}_h}{\bm{U}_j^{n,2}}-\sum_{i=1}^6 \pdiff{T_j^i}{\bm{U}_j^{n,2}} = 0, \\
  \pdiff{T_j^1}{\bm{U}_j^{n,2}} & = -\frac{h}{2}S_{n+1}^T\bm{\mu}_j^{n+1},\\
  \pdiff{T_j^2}{\bm{U}_j^{n,2}} & = -\frac{h}{2}K_{n+1/2}^T\bm{\nu}_j^{n+1}, \\
  \pdiff{T_j^4}{\bm{U}_j^{n,2}} & = \bm{M}_j^{n,2} -\frac{h}{2}S_{n+1}^T\bm{M}_j^{n,2}, \\
  \pdiff{T_j^3}{\bm{U}_j^{n,2}} & = \pdiff{T_j^5}{\bm{U}_j^{n,2}} = \pdiff{T_j^6}{\bm{U}_j^{n,2}} = 0,
\end{align*}
which we may write compactly as 
\begin{align*}
  \bm{M}_j^{n,2} + \frac{h}{2}S_{n+1}\left(\bm{\mu}_j^{n+1} + \bm{M}_j^{n,2}\right) - \frac{h}{2} K_{n+1/2}\bm{\nu}_j^{n+1} = \pdiff{{\cal J}_h}{\bm{U}_j^{n,2}}.
\end{align*}
Taking the derivative of \eqref{eq_disc-Lagrangian} with respect to $\bm{V}_j^{n,1}$ gives the set of equations
\begin{align*}
  \pdiff{\Lag_h}{\bm{V}_j^{n,1}} & = \pdiff{{\cal J}_h}{\bm{V}_j^{n,1}}-\sum_{i=1}^6 \pdiff{T_j^i}{\bm{V}_j^{n,1}} = 0, \\
  \pdiff{T_j^1}{\bm{V}_j^{n,1}} & = \frac{h}{2}K_n^T \bm{\mu}_j^{n+1},\\
  \pdiff{T_j^2}{\bm{V}_j^{n,1}} & = -\frac{h}{2}S^T_{n+1/2} \bm{\nu}_j^{n+1}, \\
  \pdiff{T_j^3}{\bm{V}_j^{n,1}} & = 0, \\
  \pdiff{T_j^4}{\bm{V}_j^{n,1}} & = \frac{h}{2}K_n^T \bm{M}_j^{n,2}, \\
  \pdiff{T_j^5}{\bm{V}_j^{n,1}} & = \bm{N}_j^{n,1} -\frac{h}{2}S_{n+1/2}^T \bm{N}_j^{n,1} , \\
  \pdiff{T_j^6}{\bm{V}_j^{n,1}} & = -\frac{h}{2}S_{n+1/2}^T \bm{N}_j^{n,2}, 
\end{align*}
which gives the condition
\begin{align*}
  \bm{N}_j^{n,1} + \frac{h}{2}S_{n+1/2}\left(\bm{\nu}_j^{n+1} + \bm{N}_j^{n,1} + \bm{N}_j^{n,2}\right) + \frac{h}{2}K_{n}\left(\bm{\mu}_j^{n+1} + \bm{M}_j^{n,2}\right) = \pdiff{{\cal J}_h}{\bm{V}_j^{n,1}}.
\end{align*}
Similarly, taking the derivative with respect to $\bm{V}_j^{n,2}$ gives 
\begin{align*}
  \pdiff{\Lag_h}{\bm{V}_j^{n,2}} & = \pdiff{{\cal J}_h}{\bm{V}_j^{n,2}}-\sum_{i=1}^6 \pdiff{T_j^i}{\bm{V}_j^{n,2}} = 0, \\
  \pdiff{T_j^1}{\bm{V}_j^{n,2}} & = \frac{h}{2}K_{n+1}^T\bm{\mu}_j^{n+1},\\
  \pdiff{T_j^2}{\bm{V}_j^{n,2}} & = -\frac{h}{2}S^T_{n+1/2} \bm{\nu}_j^{n+1}, \\
  \pdiff{T_j^4}{\bm{V}_j^{n,2}} & = \frac{h}{2}K_{n+1}^T\bm{M}_j^{n,2}, \\
  \pdiff{T_j^6}{\bm{V}_j^{n,2}} & = \bm{N}_j^{n,2}, \\
  \pdiff{T_j^3}{\bm{V}_j^{n,2}} & = \pdiff{T_j^5}{\bm{V}_j^{n,2}} = 0,
\end{align*}
giving
\begin{align*}
  \bm{N}_j^{n,2} + \frac{h}{2}S_{n+1/2}\bm{\nu}_j^{n+1} + \frac{h}{2}K_{n+1}\left(\bm{\mu}_j^{n+1} + \bm{M}_j^{n,2}\right) = \pdiff{{\cal J}_h}{\bm{V}_j^{n,2}}.
\end{align*}
In summary, the first order optimality conditions~\eqref{eq_backward} are satisfied if the following
equations hold:
\begin{gather}
  \bm{\mu}_j^n - \bm{\mu}_j^{n+1} = \bm{M}_j^{n,1} + \bm{M}_j^{n,2}, \quad \bm{\mu}_j^M =
  \pdiff{{\cal J}_h}{\bm{u}_j^M}, \label{eq:optimPreSub1}\\
  \bm{\nu}_j^n - \bm{\nu}_j^{n+1} = \bm{N}_j^{n,1} + \bm{N}_j^{n,2}, \quad \bm{\nu}_j^M  =
  \pdiff{{\cal J}_h}{\bm{v}_j^M}, \label{eq:optimPreSub2}\\
  \bm{M}_j^{n,1} + \frac{h}{2}S_n\left(\bm{\mu}_j^{n+1} + \bm{M}_j^{n,2}\right) - \frac{h}{2}
  K_{n+1/2}\left(\bm{\nu}_j^{n+1} + \bm{N}_j^{n,1} + \bm{N}_j^{n,2}\right) = \pdiff{{\cal
      J}_h}{\bm{U}_j^{n,1}}, \label{eq:optimPreSub3}\\
  \bm{M}_j^{n,2} + \frac{h}{2}S_{n+1}\left(\bm{\mu}_j^{n+1} + \bm{M}_j^{n,2}\right) - \frac{h}{2}
  K_{n+1/2}\bm{\nu}_j^{n+1} = \pdiff{{\cal J}_h}{\bm{U}_j^{n,2}}, \label{eq:optimPreSub4}\\
  \bm{N}_j^{n,1} + \frac{h}{2}S_{n+1/2}\left(\bm{\nu}_j^{n+1} + \bm{N}_j^{n,1} +
  \bm{N}_j^{n,2}\right) + \frac{h}{2}K_{n}\left(\bm{\mu}_j^{n+1} + \bm{M}_j^{n,2}\right) =
  \pdiff{{\cal J}_h}{\bm{V}_j^{n,1}},\label{eq:optimPreSub5}\\
  \bm{N}_j^{n,2} + \frac{h}{2}S_{n+1/2}\bm{\nu}_j^{n+1} + \frac{h}{2}K_{n+1}\left(\bm{\mu}_j^{n+1} +
  \bm{M}_j^{n,2}\right) = \pdiff{{\cal J}_h}{\bm{V}_j^{n,2}}.\label{eq:optimPreSub6} 
\end{gather}
We now consider the following change of variables
\begin{align}
  \bm{X}_j^{n} & = \bm{\mu}_j^{n+1} + \bm{M}_j^{n,2}, \label{eq:ChangeOfVar1}\\ 
  \bm{Y}_j^{n,1} & = \bm{\nu}_j^{n+1} + \bm{N}_j^{n,1} + \bm{N}_j^{n,2}, \label{eq:ChangeOfVar2}\\ 
  \bm{Y}_j^{n,2} & = \bm{\nu}_j^{n+1},\label{eq:ChangeOfVar3}
\end{align}
which, upon substitution into \eqref{eq:optimPreSub3}-\eqref{eq:optimPreSub6}, gives the set of equations
\begin{gather}
  \bm{M}_j^{n,1} + \frac{h}{2}S_n\bm{X}_j^{n} - \frac{h}{2} K_{n+1/2}\bm{Y}_j^{n,1} = \pdiff{{\cal
      J}_h}{\bm{U}_j^{n,1}}, \label{eq:optimPostSub3}\\
  \bm{M}_j^{n,2} + \frac{h}{2}S_{n+1}\bm{X}_j^{n} - \frac{h}{2} K_{n+1/2}\bm{Y}_j^{n,2} =
  \pdiff{{\cal J}_h}{\bm{U}_j^{n,2}}, \label{eq:optimPostSub4}\\
  \bm{N}_j^{n,1} + \frac{h}{2}S_{n+1/2}\bm{Y}_j^{n,1} + \frac{h}{2}K_{n}\bm{X}_j^{n} = \pdiff{{\cal
      J}_h}{\bm{V}_j^{n,1}},\label{eq:optimPostSub5}\\
  \bm{N}_j^{n,2} + \frac{h}{2}S_{n+1/2}\bm{Y}_j^{n,2} + \frac{h}{2}K_{n+1}\bm{X}_j^{n} =
  \pdiff{{\cal J}_h}{\bm{V}_j^{n,2}}. \label{eq:optimPostSub6}
\end{gather}

By adding \eqref{eq:optimPostSub3}-\eqref{eq:optimPostSub4},
\begin{align}
  \bm{M}_j^{n,1} + \bm{M}_j^{n,2}  = -\frac{h}{2}\left[\left(S_n + S_{n+1}\right)\bm{X}_j^{n} - K_{n+1/2}\left(\bm{Y}_j^{n,1} + \bm{Y}_j^{n,2} \right)\right]+ \pdiff{{\cal J}_h}{\bm{U}_j^{n,1}} + \pdiff{{\cal J}_h}{\bm{U}_j^{n,2}} \label{eq:sub1}.
\end{align}
Similarly, by adding \eqref{eq:optimPostSub5}-\eqref{eq:optimPostSub6},
\begin{align}
  \bm{N}_j^{n,1} + \bm{N}_j^{n,2}  = -\frac{h}{2}\left[S_{n+1/2}\left(\bm{Y}_j^{n,1} + \bm{Y}_j^{n,2} \right) + \left(K_n + K_{n+1}\right)\bm{X}_j^{n} \right]+ \pdiff{{\cal J}_h}{\bm{V}_j^{n,1}} + \pdiff{{\cal J}_h}{\bm{V}_j^{n,2}} \label{eq:sub2}.
\end{align}

Thus, \eqref{eq:optimPreSub1}-\eqref{eq:optimPreSub2} can be rewritten as 
\begin{align}
  \bm{\mu}_j^n - \bm{\mu}_j^{n+1} & = -\frac{h}{2}\left[\left(S_n + S_{n+1}\right)\bm{X}_j^{n} -
    K_{n+1/2}\left(\bm{Y}_j^{n,1} + \bm{Y}_j^{n,2} \right)\right]+ \pdiff{{\cal
      J}_h}{\bm{U}_j^{n,1}} + \pdiff{{\cal J}_h}{\bm{U}_j^{n,2}} 
  %
  \label{eq:optimPostSub1} \\
  \bm{\nu}_j^n - \bm{\nu}_j^{n+1} & = -\frac{h}{2}\left[S_{n+1/2}\left(\bm{Y}_j^{n,1} +
    \bm{Y}_j^{n,2} \right) + \left(K_n + K_{n+1}\right)\bm{X}_j^{n} \right]+ \pdiff{{\cal
      J}_h}{\bm{V}_j^{n,1}} + \pdiff{{\cal J}_h}{\bm{V}_j^{n,2}} 
  %
  \label{eq:optimPostSub2}
\end{align}
%
By combining $\bm{X}_j^{n} = \bm{\mu}_j^{n+1} + \bm{M}_j^{n,2}$ and \eqref{eq:optimPostSub4},
\begin{align}
  \bm{X}_j^n = \bm{\mu}_j^{n+1}  - \frac{h}{2}S_{n+1}\bm{X}_j^{n} + \frac{h}{2}
  K_{n+1/2}\bm{Y}_j^{n,2} + \pdiff{{\cal J}_h}{\bm{U}_j^{n,2}}.\label{eq_Xn-update}
\end{align}
Similarly, by combining $\bm{Y}_j^{n,1} = \bm{\nu}_j^{n+1} + \bm{N}_j^{n,1} + \bm{N}_j^{n,2}$ and \eqref{eq:sub2},
\begin{align}
  \bm{Y}_j^{n,1} = \bm{\nu}_j^{n+1} -\frac{h}{2}\left[S_{n+1/2}\left(\bm{Y}_j^{n,1} + \bm{Y}_j^{n,2}
    \right) + \left(K_n + K_{n+1}\right)\bm{X}_j^{n} \right]+ \pdiff{{\cal J}_h}{\bm{V}_j^{n,1}} +
  \pdiff{{\cal J}_h}{\bm{V}_j^{n,2}}.\label{eq_Yn1-update}
\end{align}
The time-stepping scheme is completed by the relation
\begin{equation}
  \bm{Y}^{n,2}=\bm{\nu}^{n+1}_j. \label{eq_Yn2-update}
\end{equation}

The scheme \eqref{eq:optimPostSub1}-\eqref{eq_Yn2-update} may be
written in the form of Lemma~\ref{lem_adjoint-bck} by defining the slopes according to \eqref{eq_kappa1}-\eqref{eq_ell2}.
This completes the proof of the lemma.




\section{Proof of Corollary 1}\label{app_cor1}
By rearranging \eqref{eq_mu} and \eqref{eq_nu},
\begin{align}
  \bm{\mu}_j^{n+1} &= \bm{\mu}_j^n + \frac{h}{2} \left(\bm{\kappa}_j^{n,1} + \bm{\kappa}_j^{n,2}\right),\label{eq_mu-forwards}\\
  \bm{\nu}_j^{n+1} &= \bm{\nu}_j^{n} + \frac{h}{2} \left( \bm{\ell}_j^{n,1} + \bm{\ell}_j^{n,2} \right).\label{eq_nu-forwards}
\end{align}
Hence, $b^\mu_1=b^\mu_2 = 1/2$ and $b^\nu_1=b^\nu_2 = 1/2$.

To express the stage variables in standard form we substitute \eqref{eq_mu-forwards} into \eqref{eq_adjoint-stageX} and define $\bm{X}_j^{n,1}=\bm{X}_j^{n,2}=\bm{X}_j^{n}$. Similarly, we substitute \eqref{eq_nu-forwards} into \eqref{eq_Y1} and \eqref{eq_Y2}, resulting in
\begin{align*}
  \bm{X}_j^{n,1} &= \bm{\mu}_j^n + \frac{h}{2} \bm{\kappa}_j^{n,1}, \\
  \bm{X}_j^{n,2} &= \bm{\mu}_j^n + \frac{h}{2} \bm{\kappa}_j^{n,1}, \\
  \bm{Y}_j^{n,1} &= \bm{\nu}_j^{n},\\
  \bm{Y}_j^{n,2} &= \bm{\nu}_j^{n} + \frac{h}{2} \left( \bm{\ell}_j^{n,1} + \bm{\ell}_j^{n,2} \right).
\end{align*}
From these relations we can identify $a^\mu_{11} = a^\mu_{21} = 1/2$ and $a^\mu_{12} = a^\mu_{22} = 0$. Furthermore, $a^\nu_{11} = a^\nu_{12} = 0$ and $a^\nu_{21} = a^\nu_{22} = 1/2$.

For the case without forcing, the formulae for the slopes,
\eqref{eq_adjoint-stageX}-\eqref{eq_Y2}, become
\begin{align}
\bm{\kappa}_j^{n,1} &= S_n \bm{X}_j^{n,1} - K_{n+1/2} \bm{Y}_j^{n,1},\label{eq_kappa1-nf}\\
\bm{\kappa}_j^{n,2} &= S_{n+1} \bm{X}_j^{n,2} - K_{n+1/2} \bm{Y}_j^{n,2},\\
\bm{\ell}_j^{n,1} &= K_{n} \bm{X}_j^{n,1} + S_{n+1/2} \bm{Y}_j^{n,1},\\
\bm{\ell}_j^{n,2} &= K_{n+1} \bm{X}_j^{n,2} + S_{n+1/2} \bm{Y}_j^{n,2}.\label{eq_ell2-nf}
\end{align}
They are consistent approximations of the time derivatives
$\dot{\bm{\mu}}(t_n)$ and $\dot{\bm{\nu}}(t_n)$, respectively. The scheme is therefore a consistent
approximation of the continuous adjoint system.

\section{Computing the gradient of the discrete objective function}\label{appendix:adjointGrad}
Given a solution that satisfies the saddle point conditions of \eqref{eq_forward} and \eqref{eq_backward}, the
gradient of ${\cal L}_h(\bm{\alpha})$ satisfies
\[
\frac{d{\cal L}_h}{d\alpha_r} = \frac{\p{\cal J}_{1h}}{\p \alpha_r}(\bm{u}, \bm{v}) + 
\frac{\p{\cal J}_{2h}}{\p \alpha_r}(\bm{U}, \bm{V}) ,\quad r=1,2,\ldots,D.
\]
The gradient of ${\cal L}_h$ with respect to $\bm{\alpha}$ only gets a contribution from the terms
in $T_j^q$ that involve the matrices $K$ and $S$. Let $S'_n = \p S/\p \alpha_r ( t_n)$ and $K'_n =
\p K/\p \alpha_r ( t_n)$. We have,
\begin{align*}
\frac{\p T^1_j}{\p \alpha_r} &=  - \frac{h}{2} \sum_{n=0}^{M-1} \left\langle S'_n\bm{U}_j^{n,1} -
K'_n\bm{V}_j^{n,1} + S'_{n+1}\bm{U}_j^{n,2} - K'_{n+1}\bm{V}_j^{n,2}, \bm{\mu}_j^{n+1} \right\rangle_2,\\ 
\frac{\p T^2_j}{\p \alpha_r} &= - \frac{h}{2} \sum_{n=0}^{M-1} \left\langle K'_{n+1/2} \left(\bm{U}_j^{n,1} + \bm{U}_j^{n,2}\right) + S'_{n+1/2}\left(
  \bm{V}_j^{n,1} +  \bm{V}_j^{n,2}\right), \bm{\nu}_j^{n+1} \right\rangle_2,\\
\frac{\p T^3_j}{\p \alpha_r} &= 0,\\
\frac{\p T^4_j}{\p \alpha_r} &= - \frac{h}{2} \sum_{n=0}^{M-1} \left\langle S'_n\bm{U}_j^{n,1} - K'_n\bm{V}_j^{n,1} + S'_{n+1}\bm{U}_j^{n,2} -
  K'_{n+1}\bm{V}_j^{n,2} , \bm{M}_j^{n,2} \right\rangle_2,\\
\frac{\p T^5_j}{\p \alpha_r} &= - \frac{h}{2} \sum_{n=0}^{M-1} \left\langle   K'_{n+1/2} \bm{U}_j^{n,1} + S'_{n+1/2} \bm{V}_j^{n,1}, \bm{N}_j^{n,1} \right\rangle_2,\\
\frac{\p T^6_j}{\p \alpha_r} &= - \frac{h}{2} \sum_{n=0}^{M-1} \left\langle   K'_{n+1/2} \bm{U}_j^{n,1} + S'_{n+1/2} \bm{V}_j^{n,1}, \bm{N}_j^{n,2} \right\rangle_2.
\end{align*}
We note that
\[
  \frac{\p (T_j^5+T_j^6)}{\p \alpha_r} =
  - \frac{h}{2} \sum_{n=0}^{M-1} \left\langle K'_{n+1/2} \bm{U}_j^{n,1} + S'_{n+1/2} \bm{V}_j^{n,1} , \bm{N}_j^{n,1} + \bm{N}_j^{n,2} \right\rangle_2.
\]
Let $\bm{X}_j^n$ and $\bm{Y}_j^{n,i}$ be defined by
\eqref{eq:ChangeOfVar1}-\eqref{eq:ChangeOfVar3}. We have,
\begin{align*}
\frac{\p T_j^4}{\p \alpha_r} &= - \frac{h}{2} \sum_{n=0}^{M-1} \left\langle S'_n\bm{U}_j^{n,1} - K'_n\bm{V}_j^{n,1} + S'_{n+1}\bm{U}_j^{n,2} -
K'_{n+1}\bm{V}_j^{n,2}, \bm{X}_j^{n} - \bm{\mu}_j^{n+1} \right\rangle_2,\\
  \frac{\p (T_j^5 + T_j^6)}{\p \alpha_r} &=
  - \frac{h}{2} \sum_{n=0}^{M-1} \left\langle  K'_{n+1/2} \bm{U}_j^{n,1} + S'_{n+1/2} \bm{V}_j^{n,1}, \bm{Y}_j^{n,1} - \bm{\nu}_j^{n+1} \right\rangle_2.
\end{align*}
Thus,
\[
  \frac{\p (T_j^1 + T_j^4)}{\p \alpha_r} = - \frac{h}{2} \sum_{n=0}^{M-1} \left\langle S'_n\bm{U}_j^{n,1} - K'_n\bm{V}_j^{n,1} + S'_{n+1}\bm{U}_j^{n,2} -
  K'_{n+1}\bm{V}_j^{n,2}, \bm{X}_j^{n} \right\rangle_2.
\]
Furthermore, from the relation \eqref{eq:ChangeOfVar3},
\begin{align*}
  \frac{\p (T_j^2 + T_j^5 + T_j^6)}{\p \alpha_r} = &- \frac{h}{2} \sum_{n=0}^{M-1} \left\langle
  K'_{n+1/2} \bm{U}_j^{n,1} + S'_{n+1/2} \bm{V}_j^{n,1}, \bm{Y}_j^{n,1}\right\rangle_2 \\ 
  &- \frac{h}{2} \sum_{n=0}^{M-1} \left\langle K'_{n+1/2} \bm{U}_j^{n,2} + S'_{n+1/2} \bm{V}_j^{n,2},\bm{Y}_j^{n,2} \right\rangle_2,
\end{align*}
We can further simplify the expressions by recognizing that $\bm{V}^{n,1}=\bm{V}^{n,2}$. By collecting the terms,
\begin{multline*}
  \frac{\p{\cal L}_h}{\p \alpha_r} =
    %
  \frac{h}{2} \sum_{j=0}^{E-1} \sum_{n=0}^{M-1} \left(  \left\langle S'_n\bm{U}_j^{n,1} + 
  S'_{n+1}\bm{U}_j^{n,2} - (K'_n + K'_{n+1}) \bm{V}_j^{n,1}, \bm{X}_j^{n}
  \right\rangle_2\, \right.\\
  + \left\langle K'_{n+1/2} \bm{U}_j^{n,1} + S'_{n+1/2} \bm{V}_j^{n,1},
  \bm{Y}_j^{n,1}\right\rangle_2\\
  \left. +\left\langle K'_{n+1/2} \bm{U}_j^{n,2} + S'_{n+1/2} \bm{V}_j^{n,1},\bm{Y}_j^{n,2}
  \right\rangle_2  \right).
\end{multline*}
This completes the proof of the lemma.

\bibliographystyle{plain}
\bibliography{quantum}
\end{document}